\documentclass[12pt, draftclsnofoot, onecolumn]{IEEEtran}

\usepackage{setspace}
\usepackage{amsmath}
\usepackage{amsfonts}
\usepackage{graphicx}
\usepackage{amsthm}
\usepackage{stfloats}
\usepackage{bm}
\usepackage{amssymb}
\usepackage{geometry}
\usepackage{color}
\newtheorem{lemma}{Lemma}
\newtheorem{theorem}{Theorem}
\newtheorem{remark}{Remark}

\begin{document}

\title{Optimal Multiuser Loading in Quantized Massive MIMO under Spatially Correlated Channels}


\author{\small
\IEEEauthorblockN{Jindan Xu$^{\ast}$,~\emph{Student Member,~IEEE}, Wei Xu$^{\ast}$,~\emph{Senior Member,~IEEE}, Fengkui Gong$^{\dag}$,~\emph{Member,~IEEE}, \\
 Hua Zhang$^{\ast}$,~\emph{Member,~IEEE}, and Xiaohu You$^{\ast}$,~\emph{Fellow,~IEEE}}\\
\IEEEauthorblockA{
$^{\ast}$National Mobile Communications Research Laboratory, Southeast University, Nanjing 210096, China\\
$^{\dag}$State Key Laboratory of Integrated Service Networks, Xidian University, Xi'an 710071, China\\
Email: \{jdxu, wxu, huazhang, xhyu\}@seu.edu.cn, fkgong@xidian.edu.cn}

\thanks{
Part of this work was presented in \cite{part} at the IEEE VTC-Fall in Toronto, Canada, Sept. 2017.
}
}

\maketitle


\begin{abstract}
Low-resolution digital-to-analog converter (DAC) has shown great potential in facilitating cost- and power-efficient implementation of massive multiple-input multiple-output (MIMO) systems.
We investigate the performance of a massive MIMO downlink network with low-resolution DACs using regularized zero-forcing (RZF) precoding.
It serves multiple receivers equipped with finite-resolution analog-to-digital converters (ADCs).
By taking the quantization errors at both the transmitter and receivers into account under spatially correlated channels, the regularization parameter for RZF is optimized with a closed-form solution by applying the asymptotic random matrix theory.
The optimal regularization parameter increases linearly with respect to the user loading ratio while independent of the ADC quantization resolution and the channel correlation.
Furthermore, asymptotic sum rate performance is characterized and a closed-form expression for the optimal user loading ratio is obtained at low signal-to-noise ratio.
The optimal ratio increases with the DAC resolution while it decreases with the ADC resolution.
Numerical simulations verify our observations.
\end{abstract}


\IEEEpeerreviewmaketitle

\section{Introduction}

Massive multiple-input multiple-output (MIMO) has gained significant attention as a candidate technique for the next generation wireless system \cite{MIMO0}-\cite{MIMO3}. In massive MIMO, a large amount of antennas equipped at base station (BS) can provide high spectral and energy efficiencies \cite{MIMO2}.
Despite these merits of massive MIMO, it suffers from a challenging issue of high cost and power consumption for applications even at the BS.
This is due to the fact that each antenna has to be driven by a separate radio-frequency (RF) chain.

Considering that the power consumption of each RF chain can decrease dramatically by reducing the resolutions of digital-to-analog converters (DACs), one of the potential solutions is to employ low-resolution DACs for downlink transmissions \cite{DAC_PSK}-\cite{DAC2}.
More specifically in \cite{DAC_PSK}-\cite{DAC_pert}, nonlinear precoding schemes were proposed for massive MIMO downlink transmission with low-resolution DACs.
In \cite{DAC_PSK}, a novel precoding technique using 1-bit DACs was presented to mitigate multiuser interference and quantization distortions.
In \cite{DAC_pokemon}, a computationally-efficient 1-bit beamforming algorithm, named POKEMON (short for PrOjected downlinK bEaMfOrmiNg), was proposed.
The authors of \cite{DAC_pert} studied perturbation methods which minimize the probability of errors at receivers in a massive MIMO downlink with 1-bit DAC.
Alternatively, linear precoding techniques were studied for the downlink transmissions using low-precision DACs \cite{DAC_downlink}-\cite{DAC2}.
A linear precoding approach was studied in \cite{DAC_downlink}, where the output data of conventional linear precoders were directly quantized by low-resolution DACs.
Considering 1-bit DACs and zero-forcing (ZF) precoder, \cite{DAC1} analyzed the system performance for massive MIMO downlink transmission.
While in \cite{DAC2}, the system performance was investigated for the downlink using  multi-bit DACs with common precoding schemes, including regularized zero-forcing (RZF) and maximal-ratio-combining (MRC) precoders.

As for uplink channels, low-resolution analog-to-digital converters (ADCs) were adopted to reduce the hardware and power cost \cite{ADC1}-\cite{ADC Mix}.
The impact of 1-bit ADC on channel estimation was considered in \cite{ADC1} and satisfactory system performance was observed in terms of both symbol error rate (SER) and mutual information.
The performance analysis was then extended in \cite{ADC2} for a multiuser relay network using mixed ADCs.
Spectral efficiency of massive MIMO with low-resolution ADCs was studied in \cite{ADC} while an extension to mixed-ADC architecture was then conducted in \cite{ADC Mix} under energy constraint.

Massive MIMO embraces an essential advantage of serving multiple users efficiently via multiuser beamforming.
The popular problem is how many users are served in order to achieve the optimal system performances.
User loading ratio, namely the ratio of the number of users simultaneously served over the antenna number, has been well recognized as a cited parameter in massive MIMO design.
It was optimized in \cite{load} that the user loading number is a complicate function of signal-to-noise ratio (SNR) and other system parameters.
Furthermore, the study in \cite{load1} found that as SNR increases, the optimal user loading decreases at low SNR, but increases at high SNR.
For a large-antenna system, the optimal user loading ratio was obtained in \cite{load2} by using the asymptotic random matrix theory.
The authors in \cite{load4} presented two user grouping methods for a frequency-division-duplexing (FDD) massive MIMO system.
While for a massive MIMO network with DAC and ADC quantizations at both sides, however, the balance of user loading of interest is still unknown for performance enhancement.

Moreover, spatial correlation at the BS side generally exists in massive MIMO since the large number of antennas could not be set far away enough from each other due to the restriction of device size.
It was revealed in \cite{cor5} that the capacity decreased about 20 percent with a separation of four wavelengths between adjacent antennas, compared to an uncorrelated MIMO channel assumed in most of the above literature.
Separable correlation model for a single-user communication system was studied in \cite{cor6}.
A more general correlation model, i.e., Unitary-Independent-Unitary (UIU) model, was introduced in \cite{cor7}.
By studying the distribution of the eigenvalues of a correlated MIMO channel matrix, closed-form approximations of asymptotic channel capacity were derived in \cite{cor3},\cite{cor4}.
Considering the transmit-side channel correlation, large system analysis for a downlink multiuser broadcast channel was conducted in~\cite{cor0}.
Using low-resolution ADCs at the BS, uplink performance of a massive MIMO network with spatial correlations was studied in \cite{cor2}.

Most of the existing studies focused on low-resolution converters at BS, i.e., few-bit DACs for downlink or ADCs for uplink.
In this paper, we investigate the joint effect of low-resolution DACs at BS along with finite-resolution ADCs at users, extending the previous conference paper in  \cite{part}.
We adopt a separable spatial correlation model, which is well-known as the Kronecker model \cite{cor6}, to characterize the spatial correlation at the BS.
At the user sides, the correlation is not considered under the assumption that single-antenna users lie far enough from each other in a rich scattering environment.
Perfect channel state information (CSI) is assumed known at the BS for analytical simplicity, considering there are already a number of studies, i.e., \cite{CE_1}-\cite{CE_3}, specifically dealing with the problem of channel estimation.
For performance enhancement, we first optimize the regularization parameter for RZF precoding design using the asymptotic random matrix theory.
Given the optimal RZF precoder, we then analyze the asymptotic sum rate performance of the massive MIMO system.
By maximizing the sum rate per antenna, the optimal user loading ratio is further derived.
Especially for low SNRs, deemed as an important scenario in massive MIMO, we obtain a closed-form solution of the optimal user loading ratio.
Engineering insights for the optimal ratio are accordingly concluded.

The rest of this paper is structured as follows. Quantization and channel models are introduced in Section~\uppercase\expandafter{\romannumeral2}.
In Section~\uppercase\expandafter{\romannumeral3}, we derive the optimal regularization parameter, asymptotic achievable rate, and the optimal user loading ratio under spatially correlated channels. Section~\uppercase\expandafter{\romannumeral4} analyzes a special case of uncorrelated channels and gives a closed-form solution of the optimal user loading for low SNR.
Section~\uppercase\expandafter{\romannumeral5} presents simulation results.
Conclusions are drawn in Section~\uppercase\expandafter{\romannumeral6}.

\emph{Notations}: $\mathbf{A}^T$, $\mathbf{A}^*$, and $\mathbf{A}^H$ represent the transpose, conjugate, and conjugate transpose of $\mathbf{A}$, respectively.
$\textrm{Tr}\{\mathbf{A}\}$ denotes the trace of $\mathbf{A}$ and $\textrm{diag}(\mathbf{A})$ keeps only the diagonal entries of $\mathbf{A}$.
$\mathbb{E}\{\cdot\}$ is the expectation operator. $\xrightarrow{a.s.}$ denotes the almost sure convergence.

\begin{figure*}[tb]
\centering\includegraphics[width=0.97\textwidth]{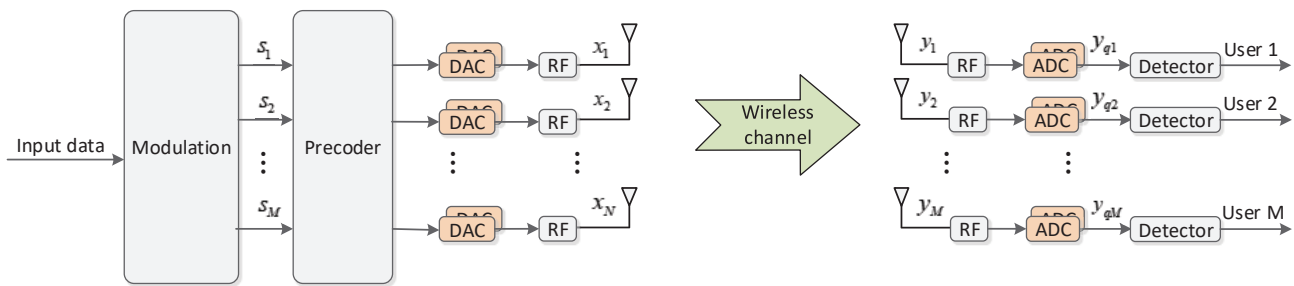}
\caption{Block diagram of multiuser massive MIMO downlink.}
\label{block}
\end{figure*}

\section{System Model}

We consider a multiuser massive MIMO downlink network as illustrated in Fig. \ref{block}. The BS equips $N$ antennas and it simultaneously serves $M$ users.
In order to reduce power consumption for driving the massive antenna array, each antenna at the BS connects to a pair of low-resolution DACs to separately process real and imaginary parts of complex signals.
Each user equips a single antenna connected to a pair of finite-resolution ADCs.
Although it is in general difficult to exactly characterize the nonlinear operation of both DACs and ADCs, the well-known Bussgang theorem \cite{Bus1},\cite{Bus2} is widely used to approximately decompose the quantized data into two uncorrelated parts.
One part represents the linearly distorted signal and the other is quantization noise.
Let $\mathbf{s}$ be the Gaussian source data with $\mathbb{E}\{\mathbf{s}\mathbf{s}^{H}\}=\mathbf{I}_{M}$ and $\mathbf{P}\in \mathbb{C}^{N\times M}$ denote the precoding matrix with $\textrm{Tr}\{\mathbf{P}\mathbf{P}^{H}\}=P$ where $P$ amounts to the transmit power budget.
By applying the Bussgang theorem, the transmit signal after DAC conversion can be modeled as \cite{ADC3}
\begin{equation}
\label{DAC}
\mathbf{x}=\sqrt{1-\rho_{DA}}\mathbf{Ps}+\mathbf{n}_{DA},
\end{equation}
where $\rho_{DA}\in (0,1)$ denotes the distortion factor which reflects the number of DAC quantization bits $b_{DA}$, and $\mathbf{n}_{DA}$ represents the quantization noise which satisfies
\begin{equation}
\begin{aligned}
\label{cor_DA}
\mathbb{E}\{\mathbf{n}_{DA}\mathbf{n}_{DA}^H\}&=\rho_{DA}\mathbb{E}\big\{\textrm{diag}\big(\mathbf{Ps}\mathbf{s}^H\mathbf{P}^H\big)\big\}
=\rho_{DA}\textrm{diag}\big(\mathbf{P}\mathbf{P}^H\big).
\end{aligned}
\end{equation}
As in most existing works, e.g., \cite{ADC2}, \cite{Bus2}-\cite{Mixed}, we consider the optimal non-uniform quantization model since it provides a tractable and effective way of well characterizing the quantization performance.
For uniform quantizers, the quantization model is more involved \cite{DAC2}.
For performance analysis, it has been shown that the resulting performance difference is marginal, by applying the two models especially for popular quantization levels.

At user sides, the stacked received vector, $\mathbf{y}\in \mathbb{C}^{M\times 1}$, of all $M$ users equals
\begin{equation}
\label{y}
\mathbf{y}=\mathbf{H}\mathbf{x}+\mathbf{n},
\end{equation}
where $\mathbf{n}\sim \mathcal{CN}(\mathbf{0},\sigma_n^2\mathbf{I}_{M})$ represents the additive white Gaussian noise (AWGN), and $\mathbf{H}\in \mathbb{C}^{M\times N}$ denotes the downlink channel with spatial correlation among transmit antennas, which can be further expressed as
\begin{equation}
\label{H}
\mathbf{H}=\mathbf{D}\tilde{\mathbf{H}}\mathbf{R}^{\frac{1}{2}},
\end{equation}
where $\mathbf{D}\in \mathbb{C}^{M\times M}$ is a diagonal matrix and the diagonal entries denote large-scale fadings, the entries of $\tilde{\mathbf{H}}\in \mathbb{C}^{M\times N}$ are assumed independent and identically distributed complex Gaussian random variables with zero mean and unit variance, and $\mathbf{R}$ denotes a transmit correlation matrix \cite{cor1}.
Assume that power control based on statistical channel-inverse is conducted to compensate large-scale fadings of different users, as done in \cite{Power}.
We consider a normalized channel model with $\mathbf{D}=\mathbf{I}_M$.
Typically, $\mathbf{R}$ is Hermitian and positive definite with $\textrm{Tr}\{\mathbf{R}\}=N$.
We assume that each user experiences the same transmit correlation since it heavily relies on the antenna array and scattering distribution at BS \cite{cor0},\cite{cor2}.
Similar to DACs, the output data vector after the ADC quantization can be decomposed according to the Bussang Theorem, as follows
\begin{equation}
\label{ADC}
\mathbf{y}_{q}=(1-\rho_{AD})\mathbf{y}+\mathbf{n}_{AD},
\end{equation}
where $\rho_{AD}\in (0,1)$ is the distortion factor of ADC quantization, and $\mathbf{n}_{AD}$ is the quantization error, which can be expressed as
\begin{equation}
\label{cor_AD}
\mathbb{E}\big\{\mathbf{n}_{AD}\mathbf{n}_{AD}^H\big\}=\rho_{AD}(1-\rho_{AD})\mathbb{E}\big\{\textrm{diag}\big(\mathbf{y}\mathbf{y}^H\big)\big\}.
\end{equation}
Compared to DACs, an additional scalar factor, $\sqrt{1-\rho_{AD}}$, is multiplied for ADCs and we have $\mathbb{E}\{\textrm{Tr}\{\mathbf{y}_q\mathbf{y}_q^{H}\}\}=(1-\rho_{AD})\mathbb{E}\{\textrm{Tr}\{\mathbf{y}\mathbf{y}^{H}\}\}$.
Power normalization is not a necessity here at the receiver because it does not affect the equivalent signal-to-interference-quantization-and-noise ratio (SIQNR) \cite{ADC3}.
The minimum distortion parameters for the optimal quantization has been studied in \cite{rho} and the typical values of $\rho_{DA}$ and $\rho_{AD}$ corresponding to various quantization resolutions are exemplified in Table \ref{table_rho}.
In particular for the special resolution of 1-bit quantization, the author in \cite{Bus1} has derived that $\rho_{DA} (\rho_{AD})=1-\frac{2}{\pi}$, which is exactly the same as in Table~\ref{table_rho} for this special case.
In the condition of moderate to high-resolution quantizations, we have $\rho_{AD}\approx\frac{\pi\sqrt{3}}{2}2^{-2b_{AD}}$ \cite{Bus2}.
The calculated values for $b_{AD}\geq3$ are verified, approximately the same as these in Table~\ref{table_rho}.

\begin{table}[tb]
\centering
\caption{Values of quantization distortion factors $\rho_{DA}$ and $\rho_{AD}$~\cite{rho}}
\label{table_rho}
\begin{IEEEeqnarraybox}[\IEEEeqnarraystrutmode\IEEEeqnarraystrutsizeadd{2pt}{1pt}]{v/c/v/c/v/c/v/c/v/c/v/c/v}
\IEEEeqnarrayrulerow\\
& \mbox{\!$b_{DA} (b_{AD})$ \!} && \!\mbox{\!1\!} && \mbox{\!2\!} && \mbox{\!3\!} && \mbox{\!4\!} && \mbox{\!5\!} &\\
\IEEEeqnarrayseprow[1pt]\\
\IEEEeqnarrayrulerow\\
\IEEEeqnarrayseprow[4pt]\\
& \mbox{\!$\rho_{DA} (\rho_{AD})$\!}  && \!\mbox{\!0.3634\!} && \mbox{\!0.1175\!} && \mbox{\!0.03454\!} && \mbox{\!0.009497\!} && \mbox{\!0.002499\!} &\IEEEeqnarraystrutsize{0pt}{0pt}\\
\IEEEeqnarrayseprow[3pt]\\
\IEEEeqnarrayrulerow
\end{IEEEeqnarraybox}
\end{table}

Now substituting \eqref{DAC} and \eqref{y} into \eqref{ADC}, the received data vector after ADC quantization equals
\begin{equation}
\begin{aligned}
\label{yq}
\mathbf{y}_{q}=&(1-\rho_{AD})\sqrt{1-\rho_{DA}}\mathbf{HPs}+(1-\rho_{AD})\mathbf{Hn}_{DA}
+\mathbf{n}_{AD}+(1-\rho_{AD})\mathbf{n}.
\end{aligned}
\end{equation}
Let $\mathbf{h}_{k}^{T}$ and $\mathbf{p}_{k}$ denote the $k$th row of $\mathbf{H}$ and the $k$th column of $\mathbf{P}$, respectively. Then, the received signal of the $k$th user can be extracted from \eqref{yq} as
\begin{align}
\label{yk}
y_{k} = &(1-\rho_{AD})\sqrt{1-\rho_{DA}}\mathbf{h}_{k}^{T}\mathbf{p}_k s_k
 +(1-\rho_{AD})\sqrt{1-\rho_{DA}}\sum_{j\neq k}\mathbf{h}_{k}^{T} \mathbf{p}_j s_j
\nonumber
\\
& +(1-\rho_{AD})\mathbf{h}_{k}^{T}\mathbf{n}_{DA}+n_{AD,k}+(1-\rho_{AD})n_{k},
\end{align}
where $s_k$, $n_{AD,k}$, and $n_{k}$ are, respectively, the $k$th elements of $\mathbf{s}$, $\mathbf{n}_{AD}$, and $\mathbf{n}$. In \eqref{yk}, the first summation term is the desired signal of user $k$ while the second term represents multiuser interference. The third and forth terms come from the DAC and ADC quantization distortions, respectively.
Then, the equivalent SIQNR of the $k$th user, $\gamma_k$, can be readily characterized as
\begin{equation}
\label{SIQNR}
\gamma_k=\frac{(1-\rho_{AD})^2(1-\rho_{DA})\left|\mathbf{h}_{k}^{T}\mathbf{p}_k\right|^2}
{(1-\rho_{AD})^2(1-\rho_{DA})\sum\limits_{j\neq k}\left|\mathbf{h}_{k}^{T}\mathbf{p}_j\right|^2+(1\!-\!\rho_{AD})^2\mathbf{h}_{k}^{T}\mathbb{E}\left\{\mathbf{n}_{DA}\mathbf{n}_{DA}^H\right\}\mathbf{h}_{k}^{*}
+\mathbb{E}\left\{|n_{AD,k}|^2\right\}\!+\!(1\!-\!\rho_{AD})^2\sigma_n^2}.
\end{equation}

\section{Asymptotic Performance Analysis}

In this section, we analyze the performance of the massive MIMO systems using the asymptotic random matrix theory \cite{matrix}, assuming both $N$ and $M$ simultaneously grow large while the user loading ratio, $\beta=\frac{M}{N}$, remains invariant.
Transmit-side channel correlation is considered and the impacts of both DAC and ADC quantizations are analyzed.

\subsection{Optimal Regularization Parameter}

Without loss of generality, we first characterize the asymptotic behavior of SIQNR at user $k$. Consider a RZF precoder as
\setcounter{equation}{9}
\begin{equation}
\label{rzf}
\mathbf{P}=c\big(\mathbf{H}^H\mathbf{H}+\alpha\mathbf{I}_N\big)^{-1}\mathbf{H}^H,
\end{equation}
where $\alpha$ is the regularization parameter and $c$ is a constant guaranteeing the power constraint as
\begin{equation}
\label{c}
c=\sqrt{\frac{P}{\textrm{Tr}\left\{\mathbf{H}\big(\mathbf{H}^H\mathbf{H}+\alpha\mathbf{I}_N\big)^{-2}\mathbf{H}^H\right\}}}.
\end{equation}
Note that a typical value of $\alpha$ is $\frac{\sigma_n^2M}{P}$ in conventional MIMO systems \cite{load1},\cite{load3}.
For the massive MIMO with low-resolution DACs at BS and finite-resolution ADCs at users, however, $\alpha$ can be further optimized as will be discussed later.
Before that, we present the asymptotic behavior of SIQNR with a Toeplitz correlation matrix as given in the following theorem.

\begin{theorem}
\label{theorem_gamma}
The asymptotic SIQNR, $\gamma$, can be expressed as
\begin{equation}
\begin{aligned}
\label{asy_SIQNR_c}
\gamma = \frac{(1-\rho_{AD})(1-\rho_{DA})\xi\left[E_{22}+\frac{\rho}{\beta}(1+\xi)^2E_{12}\right]\gamma_0}
{\rho_{AD}(1-\rho_{DA}) \xi\left[E_{22}+\frac{\rho}{\beta}(1+\xi)^2E_{12}\right]\gamma_0+\rho_{DA}(1+\xi)^2E_{12}\gamma_0+(1-\rho_{DA})E_{22}\gamma_0+(1+\xi)^2E_{12}},
\end{aligned}
\end{equation}
where $\gamma_0=\frac{P}{\sigma_n^2}$ denotes the average system SNR,
$\xi$ is defined as the unique solution of the following function
\begin{align}
\label{xi}
\xi
&=\mathbb{E}_{\lambda} \left \{ \frac{\lambda(1+\xi)}{\rho(1+\xi)+\beta\lambda} \right \},
\end{align}
where $\rho\triangleq\frac{\alpha}{N}$ is defined as the normalized regularization parameter, $\lambda$ denotes the eigenvalue of $\mathbf{R}$,
and $E_{ij}$ is defined as
\begin{align}
\label{E}
E_{ij} \triangleq \mathbb{E}_{\lambda} \left \{ \frac{\lambda^i}{[\rho(1+\xi)+\beta\lambda]^j} \right\}.
\end{align}
\end{theorem}
\begin{proof}
See Appendix~C.
\end{proof}

From \eqref{asy_SIQNR_c}, $\gamma$ is affected by distortion parameters $\rho_{DA}$ and $\rho_{AD}$ due to quantizations.
In particular, $\gamma$ decreases with both $\rho_{DA}$ and $\rho_{AD}$ because $\frac{\partial \gamma}{\partial \rho_{DA}}<0$ and $\frac{\partial \gamma}{\partial \rho_{AD}}<0$.
Since $\rho_{DA}$ and $\rho_{AD}$ are distortion factors decreasing with quantization resolution, we say that higher-resolution DAC or ADC (smaller $\rho_{DA}$ or $\rho_{AD}$) achieves larger $\gamma$, which corresponds to better performance.
The impact of spatial correlation lies in the values of $\xi$, $E_{12}$, and $E_{22}$, which depends on the eigenvalues of $\mathbf{R}$.

Note that the normalized regularization parameter, $\rho$, which is usually set as $\frac{\beta}{\gamma_0}$ for MIMO systems in existing literature, plays an important role in performance improvement with RZF.
Now, we are ready to optimize $\rho$ for rate maximization considering the impacts of channel correlation, low-resolution DACs, and finite-resolution ADCs.
The optimization problem follows
\begin{equation}
\begin{aligned}
\label{max}
\max\limits_{\rho}~~~\gamma.
\end{aligned}
\end{equation}
This optimization problem is equivalent to maximize the sum rate of all users, since the instantaneous SIQNR of each user converges to the same $\gamma$ in \eqref{asy_SIQNR_c}.
The closed-form solution to \eqref{max} is given in Lemma~\ref{lemma_rho}.

\begin{lemma}
\label{lemma_rho}
The optimal regularization parameter, $\rho^*$, is obtained as
\begin{equation}
\label{rho_opt}
\begin{aligned}
\rho^*=\frac{ \left( \rho_{DA} \gamma_0 +1 \right) \beta}{(1-\rho_{DA})\gamma_0}.
\end{aligned}
\end{equation}
\end{lemma}

\begin{proof}
In order to solve the problem in \eqref{max}, we use \eqref{asy_SIQNR_c} and set $\frac{\partial \gamma}{\partial \rho}=0$, which yields
\begin{equation}
\begin{aligned}
\label{gamma_d}
&\tilde{E}\left[\!\left(\! \rho_{DA} \gamma_0\!+\!1 \!\right)  \frac{\partial }{\partial \rho} [(\!1\!+\!\xi\!)^2E_{12}]\!+\!(1\!-\!\rho_{DA})\gamma_0 \frac{\partial E_{22}}{\partial \rho} \!\right]
\!-\!\frac{\partial \tilde{E}}{\partial \rho}\left[\left(\rho_{DA}\gamma_0\!+\!1\right)(1+\xi)^2E_{12}\!+\!(1\!-\!\rho_{DA})E_{22}\gamma_0\right]
\!=\!0,
\end{aligned}
\end{equation}
where we define $\tilde{E}=\xi\left[E_{22}+\frac{\rho}{\beta}(1+\xi)^2E_{12}\right]$ for notational brevity and it follows
\begin{align}
\label{A_d}
\frac{\partial \tilde{E}}{\partial \rho} =&  \frac{\partial \xi}{\partial \rho}\left[E_{22}+\frac{\rho}{\beta}(1+\xi)^2E_{12}\right]
 +\xi \left[ \frac{\partial E_{22}}{\partial \rho} +  \frac{(1+\xi)^2E_{12}}{\beta}   + \frac{\rho}{\beta} \frac{\partial }{\partial \rho} \left[(1+\xi)^2E_{12} \right]\right]
\nonumber
 \\
= & \xi \frac{\partial E_{22}}{\partial \rho} + \frac{\xi \rho}{\beta} \frac{\partial }{\partial \rho} \left[(1+\xi)^2E_{12}\right],
\end{align}
where \eqref{xi_d} and \eqref{xi0} in Appendix~C are utilized.
Then, by substituting \eqref{A_d} into \eqref{gamma_d} and after some basic manipulations, Lemma~\ref{lemma_rho} is obtained.
\end{proof}

Comparing \eqref{rho_opt} with the conventional typical value of $\rho$, we can rewrite $\rho^*=\frac{\beta}{\tilde{\gamma}}$,
where $\tilde{\gamma} \triangleq \frac{(1-\rho_{DA})\gamma_0}{\rho_{DA}\gamma_0+1}$ can be regarded as an equivalent system SNR affected by low-resolution DACs.
Analogous to the typical $\rho=\frac{\beta}{\gamma_0}$, the derived $\rho^*$ replaces the original SNR $
\gamma_0$ with the equivalent one, i.e., $\tilde{\gamma}$.

\textbf{Remark 1.}
It is interesting that $\rho^*$ in \eqref{rho_opt} is independent of correlation matrix $\mathbf{R}$, which implies that \emph{the optimal RZF precoding is not affected by the spatial correlation at the BS.}
Besides, $\rho^*$ is independent of $\rho_{AD}$ which implies that \emph{the resolution of receiving ADCs does not affect the optimal RZF precoding.} This looks promising since it allows the BS to determine the optimal RZF parameter without knowing the ADC setups at user sides.
From \eqref{rho_opt}, $\rho^*$ increases with $\rho_{DA}$ while it decreases with $b_{DA}$.

\vspace{0.2cm}
\subsection{Optimal User Loading Ratio}

In general, serving more users directly increases instantaneous sum rate while leading to heavier pilot overhead for channel estimation.
It is interesting to optimize the user loading ratio, $\beta$, by characterizing the rate performance within a period no longer than coherence interval.

Substituting \eqref{rho_opt} into \eqref{asy_SIQNR_c}, the maximized SIQNR of each user can be expressed as
\begin{equation}
\label{SIQNR_opt}
\gamma^*=\frac{(1-\rho_{AD})\xi^*}{1+\rho_{AD} \xi^*},
\end{equation}
where $\xi^*$ is computed from \eqref{xi} by substituting $\rho^*$ in \eqref{rho_opt}.
On one hand, the effect of ADC quantization is evident. The more quantization bits ADCs provide, the smaller $\rho_{AD}$ is and the larger $\gamma^*$ will be. One the other hand, the impact of low-resolution DAC lies in $\rho^*$ as indicated in \eqref{rho_opt}.
Besides, the spatial correlation affects the asymptotic performance through~$\xi^*$.

By applying the assumption of the worst-case Gaussian interference, the asymptotic achievable rate of each user equals
\begin{equation}
\label{rate}
R=\log_2\left(1+\gamma^*\right)=\log_2\frac{1+\xi^*}{1+\rho_{AD} \xi^*}.
\end{equation}
$R$ in (20) represents the achievable rate of each user per channel use without considering the pilot overhead.
Assume that each user needs an overhead of $\tau$ to transmit \emph{precoded pilot} signal for user detection within a coherence interval $T$. The achievable sum rate per antenna including the pilot overhead can then be evaluated from \eqref{rate} as
\begin{equation}
\begin{aligned}
\label{rate_avg}
\overline{R}=\beta \left(1-\frac{M\tau}{T}\right)R
=\beta (1-\eta\beta)\log_2 \frac{1+\xi^*}{1+\rho_{AD} \xi^*},
\end{aligned}
\end{equation}
where $\eta \triangleq \frac{N\tau}{T}$.
Then, we can acquire the optimal user loading ratio, $\beta^*$, by maximizing $\overline{R}$, which yields
\begin{equation}
\begin{aligned}
\label{max_R}
\beta^*= \arg~\max\limits_{\beta}~\overline{R}.
\end{aligned}
\end{equation}
Since $\overline{R}$ in \eqref{rate_avg} is a continuous function w.r.t $\beta\in(0,1)$, $\beta^*$ is obtained when the derivative of $\overline{R}$ w.r.t. $\beta$ equals zero, or tends to the critical values that $\beta^*\rightarrow 0$ or $\beta^*\rightarrow 1$.
Consider the scenario of low SNR with $\gamma_0\ll 1$.
Substituting $\beta\rightarrow 0$ to \eqref{rate_avg}, it is obvious that $\overline{R}\rightarrow 0$.
For $\beta\rightarrow 1$ and $\gamma_0\ll 1$, we have $\rho^*\rightarrow\frac{ 1}{(1-\rho_{DA})\gamma_0}$ by using \eqref{rho_opt}.
Substituting this into \eqref{xi} yields $\xi^*\rightarrow\mathbb{E}_{\lambda} \left \{ \lambda(1-\rho_{DA})\gamma_0 \right \}\ll 1$, which implies $\overline{R}\rightarrow 0$.
As a result, we can conclude that the maximum of $\overline{R}$ can only be obtained by solving the equation
\begin{equation}
\begin{aligned}
\label{rate_d}
\frac{\partial \overline{R} }{\partial \beta}=0.
\end{aligned}
\end{equation}
Although it is hard to find a simple closed-form solution of \eqref{rate_d} in general cases, we can resort to offline numerical methods, e.g., the bisection method, as in Section~\uppercase\expandafter{\romannumeral5}.
Since $\beta^*$ in \eqref{max_R} depends only on system parameters, lookup tables of $\beta^*$ can be constructed.

Note that in maximizing $\overline{R}$, we obtain $\rho^*$ for any value of $\beta$.
Then, we optimize $\beta$ by substituting the obtained $\rho^*$ in \eqref{rho_opt}, which can be considered as a function w.r.t $\beta$.
Therefore, it is safe to conclude that both the achieved $\rho^*$ and $\beta^*$ are globally optimal.


\vspace{0.2cm}

\section{Performance under Uncorrelated Channels}

In this section, we investigate a special case of uncorrelated channels which allows a closed-form solution to the optimal user loading ratio at low SNR.

\subsection{Asymptotic Achievable Rate}
Under the special case of uncorrelated channels, i.e., $\mathbf{R}=\mathbf{I}$ in \eqref{H}, the value of $\xi$ in \eqref{xi} can be further simplified. It reduces to the unique positive solution of the following function
\begin{align}
\label{g}
\xi=\frac{(1+\xi)}{\rho(1+\xi)+\beta}.
\end{align}
Note that the solution to $\xi$ is a function of $\beta$ and $\rho$.
By solving the quadratic equality in \eqref{g}, the positive one of its solutions, referred to as $\xi_{UC}$, equals
\begin{equation}
\label{g_d}
\xi_{UC}(\beta,\rho)\!=\!\frac{1}{2}\left[\sqrt{\frac{(1\!-\!\beta)^2}{\rho^2}\!+\!\frac{2(1\!+\!\beta)}{\rho}+1}\!+\!\frac{1\!-\!\beta}{\rho}\!-\!1\!\right].
\end{equation}
From the definition of $E_{ij}$ in \eqref{E}, we have $E_{12}=E_{22}$ with $\lambda\equiv1$ for an uncorrelated massive MIMO channel. Thus, the asymptotic SIQNR per user in \eqref{asy_SIQNR_c} can be simplified as
\begin{equation}
\begin{aligned}
\label{asy_SIQNR_uc}
\gamma_{UC} = \frac{(1-\rho_{AD})(1-\rho_{DA})\xi_{UC}\left[1+\frac{\rho}{\beta}(1+\xi_{UC})^2\right]\gamma_0}
{\rho_{AD} (1-\rho_{DA}) \xi_{UC}\left[1+\frac{\rho}{\beta}(1+\xi_{UC})^2\right]\gamma_0+\rho_{DA}(1+\xi_{UC})^2\gamma_0+(1-\rho_{DA})\gamma_0+(1+\xi_{UC})^2}.
\end{aligned}
\end{equation}

From the conclusions in Remark 1, it indicates that the optimal regularization parameter is independent of the spatial correlation.
Thus, $\rho^*$ remains the same as expressed in \eqref{rho_opt}, which maximizes $\gamma_{UC}$ in \eqref{asy_SIQNR_uc}.
One typical scenario for massive MIMO is widely acknowledged as the low SNR case. We focus on this special case with $\gamma_0 \ll 1$.
Substituting $\rho^*$ in \eqref{rho_opt} into $\xi_{UC}(\beta,\rho)$ in \eqref{g_d}, it yields
\begin{align}
\nonumber
\xi_{UC}^*(\beta,\rho^*)&=
\frac{1}{2} \sqrt{\frac{(1-\beta)^2}{\left[\frac{1}{(1-\rho_{DA})\gamma_0}\!+\!\frac{\rho_{DA}}{1-\rho_{DA}}\right]^2\beta^2}\!+\!\frac{2(1+\beta)}{\left[\frac{1}{(1-\rho_{DA})\gamma_0}\!+\!\frac{\rho_{DA}}{1-\rho_{DA}}\right]\beta}\!+\!1}
 +\frac{1-\beta}{2\left[\frac{1}{(1-\rho_{DA})\gamma_0}\!+\!\frac{\rho_{DA}}{1-\rho_{DA}}\right]\beta}-\frac{1}{2}
\nonumber
\\
&\approx (1-\rho_{DA})\frac{\gamma_0}{\beta}+\textrm{o}(\gamma_0)
\label{g_low}
,
\end{align}
where \eqref{g_low} applies the Taylor's expansion for small $\gamma_0$,
i.e., $\xi_{UC}^*=\xi_{UC}^*\Big |_{\gamma_0=0}+\frac{\partial \xi_{UC}^* }{\partial \gamma_0} \Big |_{\gamma_0=0} \gamma_0 + \textrm{o}(\gamma_0)$,
and uses $\xi_{UC}^*\Big |_{\gamma_0=0}\approx0$ and $\frac{\partial \xi_{UC}^* }{\partial \gamma_0} \Big |_{\gamma_0=0}\approx\frac{1-\rho_{DA}}{\beta}$ for small $\gamma_0$.
$\textrm{o}(\gamma_0)$ is an ignorable higher-order term w.r.t $\gamma_0$.
Replacing $\xi^*$ in \eqref{rate} with $\xi_{UC}^*$ and substituting \eqref{g_low}, the achievable rate of each user at low SNR can be rewritten as:
\begin{align}
R_{UC}
&=\log_2\frac{1+\xi_{UC}^*}{ 1+\rho_{AD} \xi_{UC}^*}
=\log_2\frac{\beta+(1-\rho_{DA})\gamma_0}{\beta+\rho_{AD} (1-\rho_{DA})\gamma_0}.
\label{rate_low}
\end{align}
We observe that a smaller value of $\rho_{AD}$, $\rho_{DA}$ or $\beta$ provides higher $R_{UC}$. It implies that high-resolution ADCs, DACs, or a small user loading ratio provides high achievable rate for each user.

Then, we characterize the rate loss caused by the finite-resolution ADCs for a fixed quantization resolution of DACs.

\begin{lemma}
\label{lemma_rareloss}
Under the uncorrelated massive MIMO channel, the normalized rate loss per energy caused by finite-resolution ADCs with fixed-bit DACs can be characterized at low SNR as
\begin{align}
\label{delta_rate_low}
\frac{\Delta R_{UC}}{\gamma_0}
\rightarrow \frac{1}{\ln 2}\rho_{AD}(1-\rho_{DA}) \frac{1}{ \beta}.
\end{align}
\end{lemma}


\begin{proof}
Given the ideal infinite-resolution ADCs, the benchmark asymptotic rate can be evaluated from \eqref{rate_low} by letting $\rho_{AD}=0$. It gives
\begin{equation}
\label{rate_0}
\tilde{R}_{UC}=\log_2\frac{\beta+(1-\rho_{DA})\gamma_0}{\beta}.
\end{equation}
Then, the normalized rate loss per energy can be characterized by subtracting \eqref{rate_low} from \eqref{rate_0} as follows
\begin{align}
\frac{\Delta R_{UC}}{\gamma_0}&= \frac{\tilde{R}_{UC}}{\gamma_0}-\frac{R_{UC}}{\gamma_0}
=\log_2 \left[1+\rho_{AD}(1-\rho_{DA}) \frac{\gamma_0}{\beta}\right]^{\frac{1}{\gamma_0}}
\rightarrow \frac{1}{\ln 2}\rho_{AD}(1-\rho_{DA}) \frac{1}{ \beta}
,
\label{Delta_R}
\end{align}
where we use the fact that $(1+x)^{\frac{1}{x}}\rightarrow \textrm{e}$ for $x\rightarrow 0$.
\end{proof}

\textbf{Remark 2.}
Given $\rho_{DA}$, $\frac{\Delta R_{UC}}{\gamma_0}$ is approximately linearly affected by $\rho_{AD}$ and $\frac{1}{\beta}$ at low SNR.
Particularly, for $\beta=\frac{1}{2}$ as set in the following simulations, we have $\frac{\Delta R_{UC}}{\gamma_0}\rightarrow 0.06345$ with $b_{DA}=1$ and $b_{AD}=3$ according to Table~\ref{table_rho}.
This implies that, given 1-bit DACs at the BS, the normalized rate loss in \eqref{delta_rate_low} caused by 3-bit ADCs at the user side is relatively negligible.
\subsection{Optimal User Loading Ratio for Uncorrelated Channels}

Under the assumption of low SNR, we are able to obtain the optimal user loading ratio in closed form by maximizing the sum rate per antenna, given in Lemma~\ref{lemma_beta}.

\begin{lemma}
\label{lemma_beta}
For an uncorrelated massive MIMO system, the problem in \eqref{max_R} has a closed-form solution for asymptotically low SNR as
\begin{equation}
\begin{aligned}
\label{beta_opt}
\beta^*\!=\!
&-\gamma_0(1+\rho_{AD})(1-\rho_{DA})
+\!\sqrt{\gamma_0^2(1\!+\!\rho_{AD})^2(1\!-\!\rho_{DA})^2\!+\!\frac{\gamma_0(1\!+\!\rho_{AD})(1\!-\!\rho_{DA})}{\eta}}.
\end{aligned}
\end{equation}
\end{lemma}

\begin{proof}
Utilizing \eqref{rate_avg}, we first rewrite the sum rate per antenna without spatial correlation as
\begin{equation}
\begin{aligned}
\label{rate_avg_uc}
\overline{R}_{UC}
=\beta (1-\eta\beta)\log_2 \frac{1+\xi_{UC}^*}{1+\rho_{AD} \xi_{UC}^*},
\end{aligned}
\end{equation}
where $\xi_{UC}^*$ is expressed in \eqref{g_low}.
It directly yields that
\begin{align}
\frac{\partial \overline{R}_{UC} }{\partial \beta}
=&(1-2\eta\beta)\log_2 \frac{1+\xi_{UC}^*}{1+\rho_{AD} \xi_{UC}^*}
+\frac{\beta (1-\eta\beta)(1-\rho_{AD})}{\ln2(\xi_{UC}^*+1)(\rho_{AD} \xi_{UC}^*+1)}\frac{\partial \xi_{UC}^* }{\partial \beta}
\nonumber
\\
\approx &(1-2\eta\beta) \log_2 \frac{1+\xi_{UC}^*}{1+\rho_{AD} \xi_{UC}^*}
-\frac{\xi_{UC}^* (1-\eta\beta)(1-\rho_{AD})}{\ln2(\xi_{UC}^*\!+\!1)(\rho_{AD} \xi_{UC}^*+1)},
\label{r_b}
\end{align}
where \eqref{r_b} comes from $\frac{\partial \xi_{UC}^* }{\partial \beta}\approx-\frac{\xi_{UC}^*}{\beta}$ according to \eqref{g_low}.
Letting $\frac{\partial \overline{R}_{UC} }{\partial \beta}=0$, we get
\begin{equation}
\begin{aligned}
\label{r_b_0}
\rho_{AD}(\xi_{UC}^*)^2+ (1+\rho_{AD})\xi_{UC}^*-\frac{\eta\beta}{1-2\eta\beta}=0,
\end{aligned}
\end{equation}
where we use the equality $\ln(1+x)=x+\textrm{o}(x)$ for small $x$.
Substituting \eqref{g_low} into \eqref{r_b_0} and dropping $\rho_{AD}(\xi_{UC}^*)^2$ since $\xi_{UC}^*\ll1$, we arrive at a quadratic equation w.r.t. $\beta$, i.e, $\eta\beta^2\!+\!2\eta\gamma_0(1\!+\!\rho_{AD})(1\!-\!\rho_{DA})\beta\!-\!(1\!+\!\rho_{AD})(1\!-\!\rho_{DA})\gamma_0\!=\!0$, which has only one positive solution given in~\eqref{beta_opt}.

\end{proof}

\textbf{Remark 3.}
On one hand, $\beta^*$ is an increasing function w.r.t $\gamma_0$.
More users should be loaded at higher SNR in order to achieve the maximum sum rate per antenna, while it increases the total pilot overhead.
On the other hand, $\beta^*$ increases with $\rho_{AD}$ while it decreases with $\rho_{DA}$.
It implies that \emph{the optimal user loading ratio increases as the ADC resolution decreases at the users or the DAC resolution increases at the BS.}

\section{Simulation Results}

In this section, we evaluate the asymptotic SIQNR, $\gamma$, the asymptotic achievable rate of each user, $R$, and the asymptotic sum rate per antenna, $\overline{R}$, via simulations.
Optimal regularization parameter $\rho^*$ is verified by numerical methods.
Moreover, the closed-form solution of $\beta^*$ in \eqref{beta_opt} is also tested at low SNRs.
Consider a well-known exponential Toeplitz correlation model as follows
\begin{align}
\label{r}
r_{ij}=\nu^{|i-j|},
\end{align}
where $\nu$ denotes the correlation coefficient and $r_{ij}~(i,j=1,2,...,N)$ represents the entry of $\mathbf{R}$ at the $i$th row and $j$th column.
We use $\nu=0.2-0.8$ for simulation in the following.
Generally, small correlation coefficient is representative in a rich scattering environment, e.g., $\nu=0.2$ for indoor scenarios.
While large values of, e.g., $\nu=0.8$ represents the scenario where the BS locates in an unobstructed environment.

\begin{figure}[tb]
\centering\includegraphics[width=0.6\textwidth]{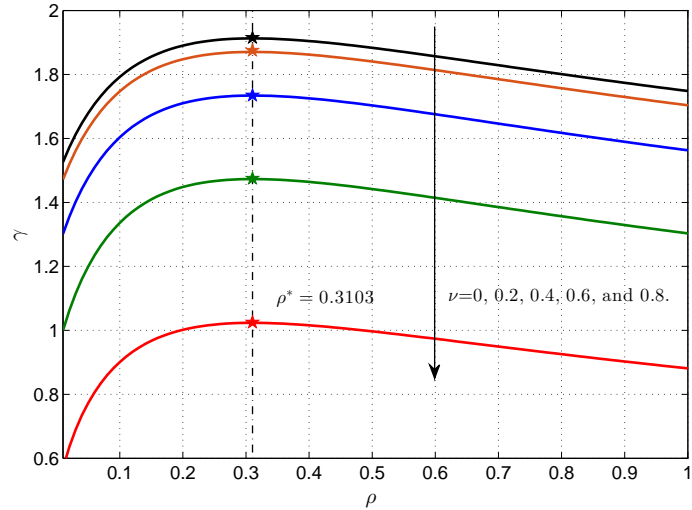}
\caption{SIQNR versus regularization parameter with various correlation coefficients ($N=64$, $M=32$, $\gamma_0=15$ dB, $b_{DA}=1$, and $b_{AD}=3$).}
\label{SIQNR_nu}
\end{figure}
\begin{figure}[tb]
\centering\includegraphics[width=0.6\textwidth]{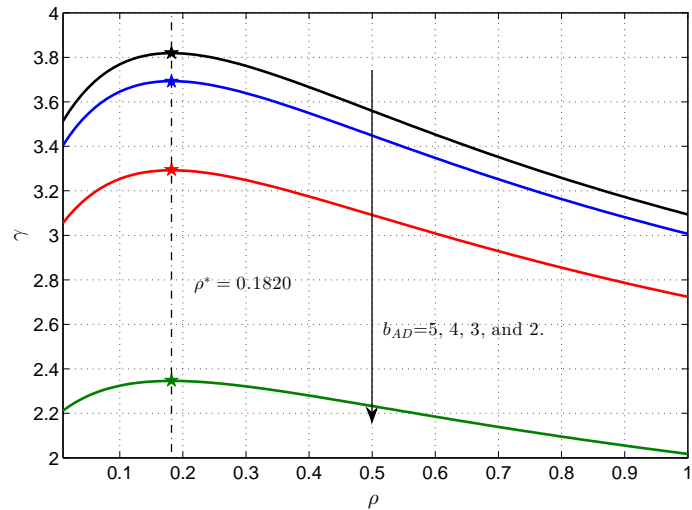}
\caption{SIQNR versus regularization parameter with various ADC resolutions ($N=64$, $M=16$, $\gamma_0=10$ dB, $b_{DA}=1$, and $\nu=0.5$).}
\label{SIQNR_AD}
\end{figure}

\begin{figure}[tb]
\centering\includegraphics[width=0.6\textwidth]{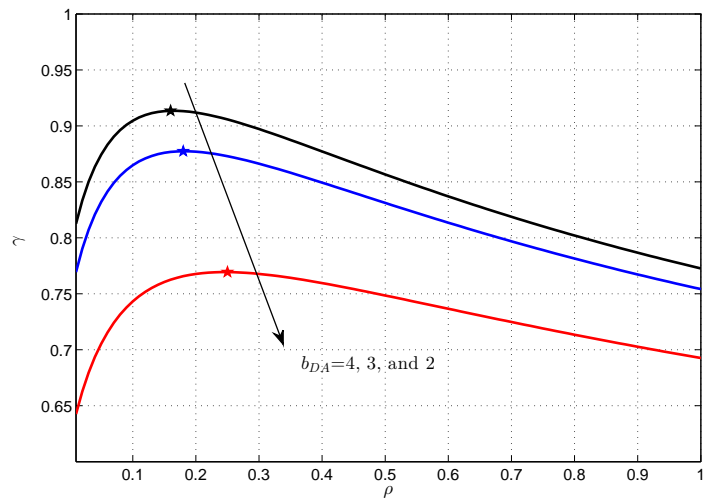}
\caption{SIQNR versus regularization parameter with various DAC resolutions ($N=32$, $M=16$, $\gamma_0=5$ dB, $b_{AD}=1$, and $\nu=0.5$).}
\label{SIQNR_DA}
\end{figure}

\begin{figure}[tb]
\centering\includegraphics[width=0.6\textwidth]{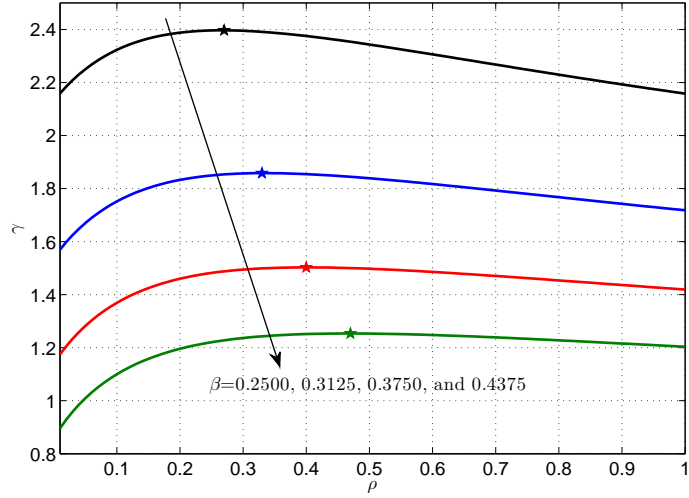}
\caption{SIQNR versus regularization parameter with various user loading ratio ($\gamma_0=5$ dB, $b_{DA}=1$, $b_{AD}=3$, $\nu=0.5$, $N=64$, and $M=16,20,24,28$).}
\label{SIQNR_beta}
\end{figure}

\begin{figure}[tb]
\centering\includegraphics[width=0.6\textwidth]{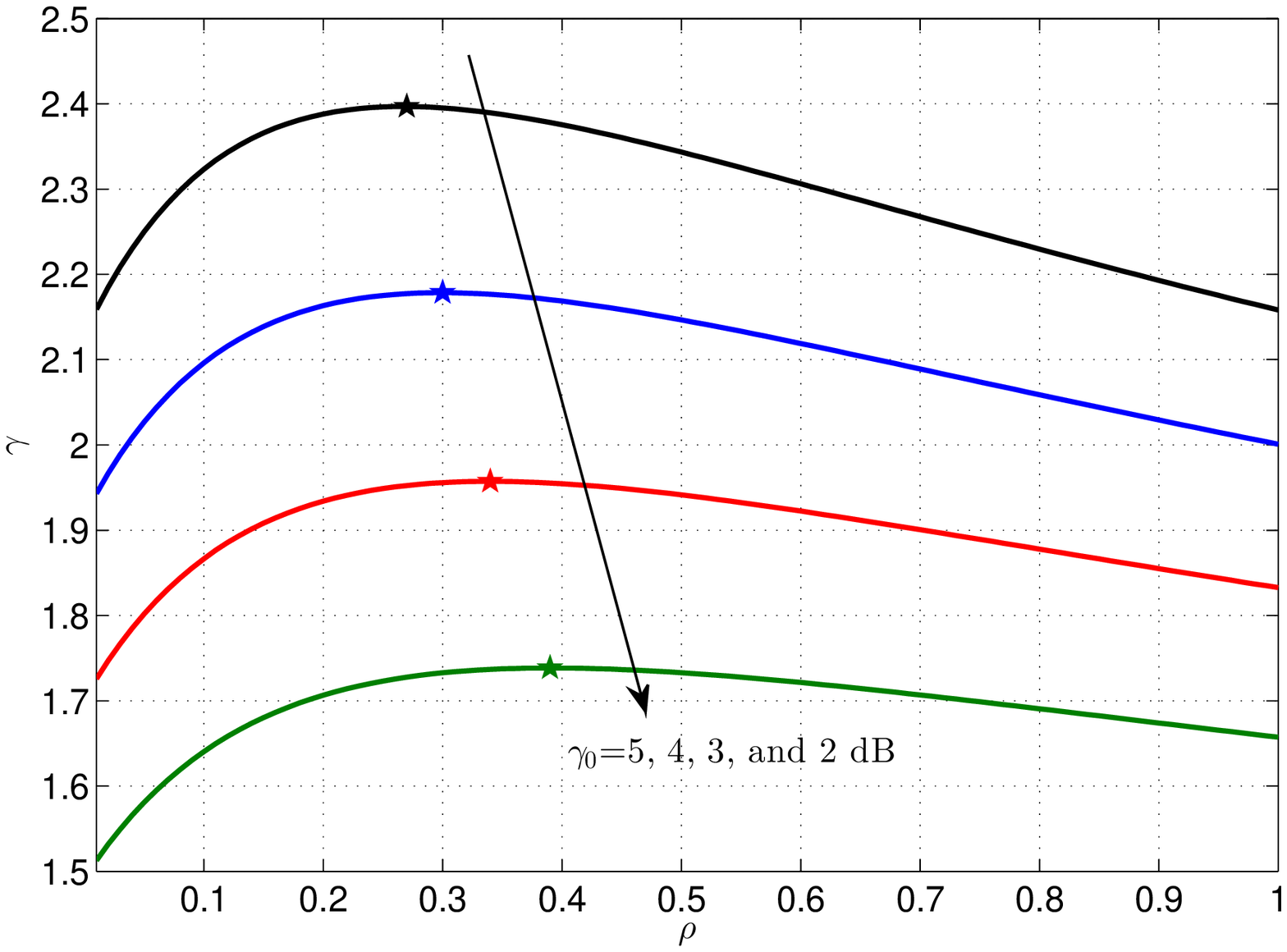}
\caption{SIQNR versus regularization parameter with various SNRs ($N=64$, $M=16$, $b_{DA}=1$, $b_{AD}=3$, $\nu=0.5$, and $\gamma_0=5, 4, 3, 2$ dB).}
\label{SIQNR_SNR}
\end{figure}

We calculate $\gamma$ in \eqref{asy_SIQNR_c} and test our derived $\rho^*$ in \eqref{rho_opt} using the values of $\xi$, $E_{12}$, and $E_{22}$ calculated in Appendix~D.
Fig.~\ref{SIQNR_nu} displays the asymptotic SIQNR with various correlation coefficients.
Obviously, $\gamma$ first increases but then decreases with increasing $\rho$.
The pentagrams mark the maximum $\gamma$ with $\rho^*$, which verifies $\rho^*=0.3103$ obtained from \eqref{rho_opt}.
It can be observed that $\rho^*$ is independent of $\nu$ as analyzed in Remark 1, which implies that the optimal RZF precoding design is independent of the antenna spatial correlation.
Fig. \ref{SIQNR_AD} shows the asymptotic SIQNR with various ADC resolutions, i.e., $b_{AD}=2-5$ bits.
$\rho^*$ is marked by pentagrams and is independent of the ADC resolution as indicated by \eqref{rho_opt}.
Fig. \ref{SIQNR_DA} displays the asymptotic SIQNR with varying DAC resolutions.
According to \eqref{rho_opt}, $\rho^*$ marked by pentagrams are $0.1644, 0.1817$, and $0.2457$ for $b_{DA}=4, 3$, and $2$, respectively.
It can be observed that $\rho^*$ increases with decreasing $b_{DA}$ as discussed in Remark 1.
Fig. \ref{SIQNR_beta} and Fig. \ref{SIQNR_SNR} show the asymptotic SIQNR with various $\beta$ and $\gamma_0$.
In Fig. \ref{SIQNR_beta}, $\rho^*=0.2669, 0.3336, 0.4003$, and $0.4671$ for $\beta=0.2500, 0.3125, 0.3750$, and $0.4375$, respectively.
We verify that $\rho^*$ proportionally increases with $\beta$.
In Fig. \ref{SIQNR_SNR}, $\rho^*=0.2669, 0.2991, 0.3395$, and $0.3905$ for $\gamma_0=5,4,3$, and $2$ dB, which verifies that $\rho^*$ decreases with $\gamma_0$.

\begin{figure}[tb]
\centering\includegraphics[width=0.6\textwidth]{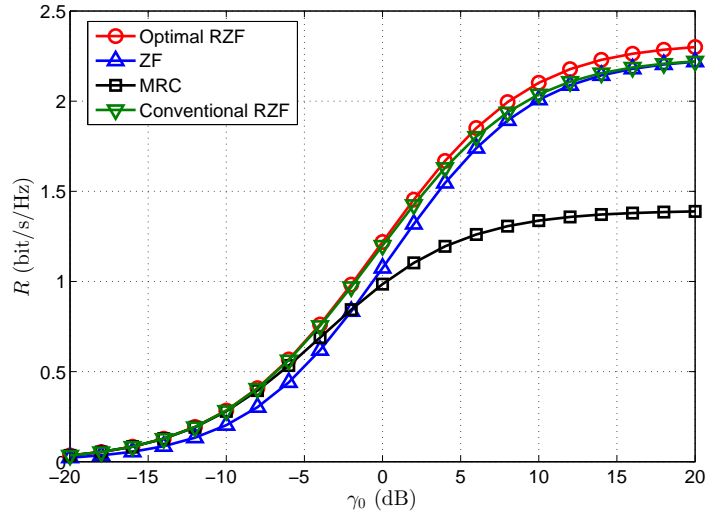}
\caption{Achievable rate versus SNR with optimal and conventional  RZF, ZF, and MRC precoders ($N=256$, $M=64$, $b_{DA}=1$, $b_{AD}=3$, and $\nu=0.5$).}
\label{rate_RZF_ZF_MRC_fig}
\end{figure}

Fig. \ref{rate_RZF_ZF_MRC_fig} compares the achievable rate per user with the optimal $\rho$ and other three typical values, i.e., $\rho=\frac{\beta}{\gamma_0}$ leading to a conventional RZF precoder, $\rho\rightarrow0$ leading to a ZF precoder, and $\rho\rightarrow \infty$ leading to an MRC precoder.
The performance gap between the optimal RZF and the conventional RZF is generally not that significant, but it increases with SNR.
This is because the effect of quantization noise becomes to dominate the performance at high SNR.
Thus, the regularization parameter is better to be optimized according to quantization distortion factors.
In addition, the optimal RZF precoding provides better rate performance than ZF precoding with $\gamma_0$ ranging from $-20$ to $20$ dB.
As for MRC precoding, it achieves almost the same rate as the optimal RZF at low SNRs since the channel noise significantly overwhelms the inter-user interference under this condition. However, the multiuser interference becomes dominating as SNR increases and thus the achievable rate with MRC precoder becomes much lower than the optimal RZF at high SNRs.

\begin{figure}[tb]
\centering\includegraphics[width=0.6\textwidth]{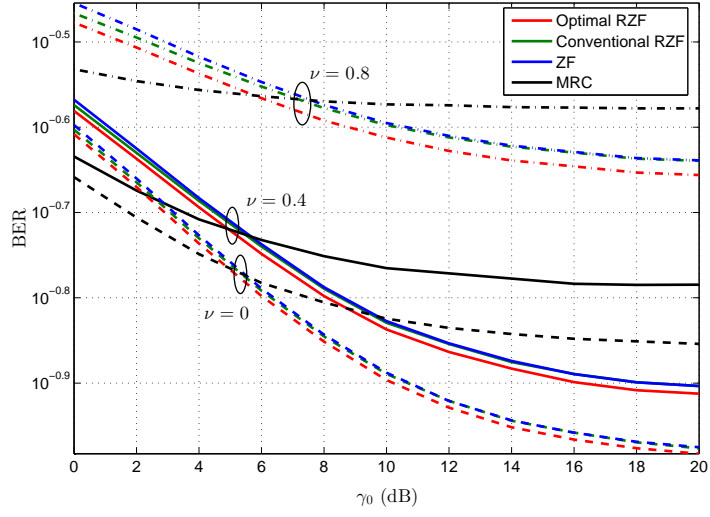}
\caption{BER versus SNR with optimal and conventional  RZF, ZF, and MRC precoders ($N=64$, $M=32$, $b_{DA}=1$, and $b_{AD}=3$).}
\label{BER_fig}
\end{figure}

Fig. \ref{BER_fig} compares the bit error rate (BER) of various precoding schemes using QPSK modulation. 
Error floor exists due to the use of 1-bit DACs and 3-bit ADCs.
Compared to conventional RZF and ZF precoders, the optimal RZF precoder improves the BER performance more significantly with larger $\nu$.
It implies that the optimal RZF precoding is more effective for spatially correlated massive MIMO channels using low-resolution DACs and finite-resolution ADCs.
Note that the BER values are high because we consider a heavily loaded scenario with $\beta=\frac{M}{N}=\frac{1}{2}$ and no channel coding is employed.
This allows us to better focus on evaluating the performance of the proposed optimal RZF precoder.

\begin{figure}[tb]
\centering\includegraphics[width=0.6\textwidth]{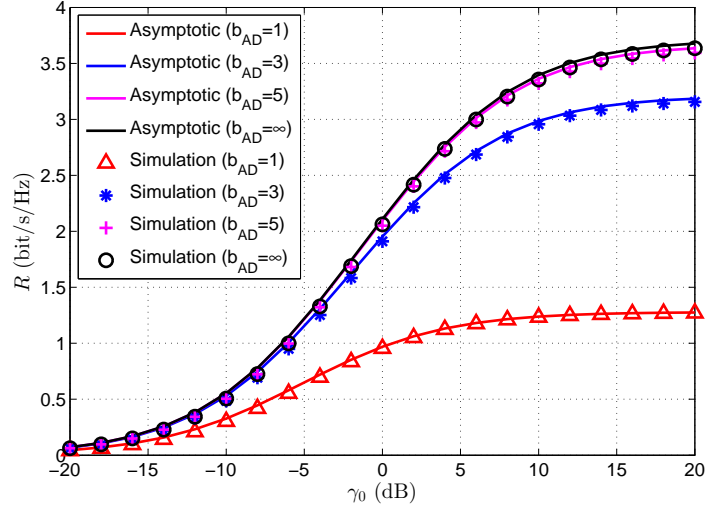}
\caption{Achievable rate of each user using 1-bit DACs and finite-resolution ADCs ($N=256, M=32$, and $\nu=0.3$).}
\label{rate_fig}
\end{figure}

Fig. \ref{rate_fig} shows the achievable rate of each user affected by finite-resolution ADCs.
Solid lines correspond to our derived asymptotic expression in \eqref{rate} and it matches well with the simulation results. Note that $R$ converges with increasing SNR since only one-bit information can be transmitted by each antenna with 1-bit DACs.
Besides, higher ADC resolution $b_{AD}$ provides larger $R$ and 5-bit ADCs achieve almost the same rate as infinite-resolution ones.
Note that at low SNRs, the rate loss caused by ADC quantization is marginal.

\begin{figure}[tb]
\centering\includegraphics[width=0.6\textwidth]{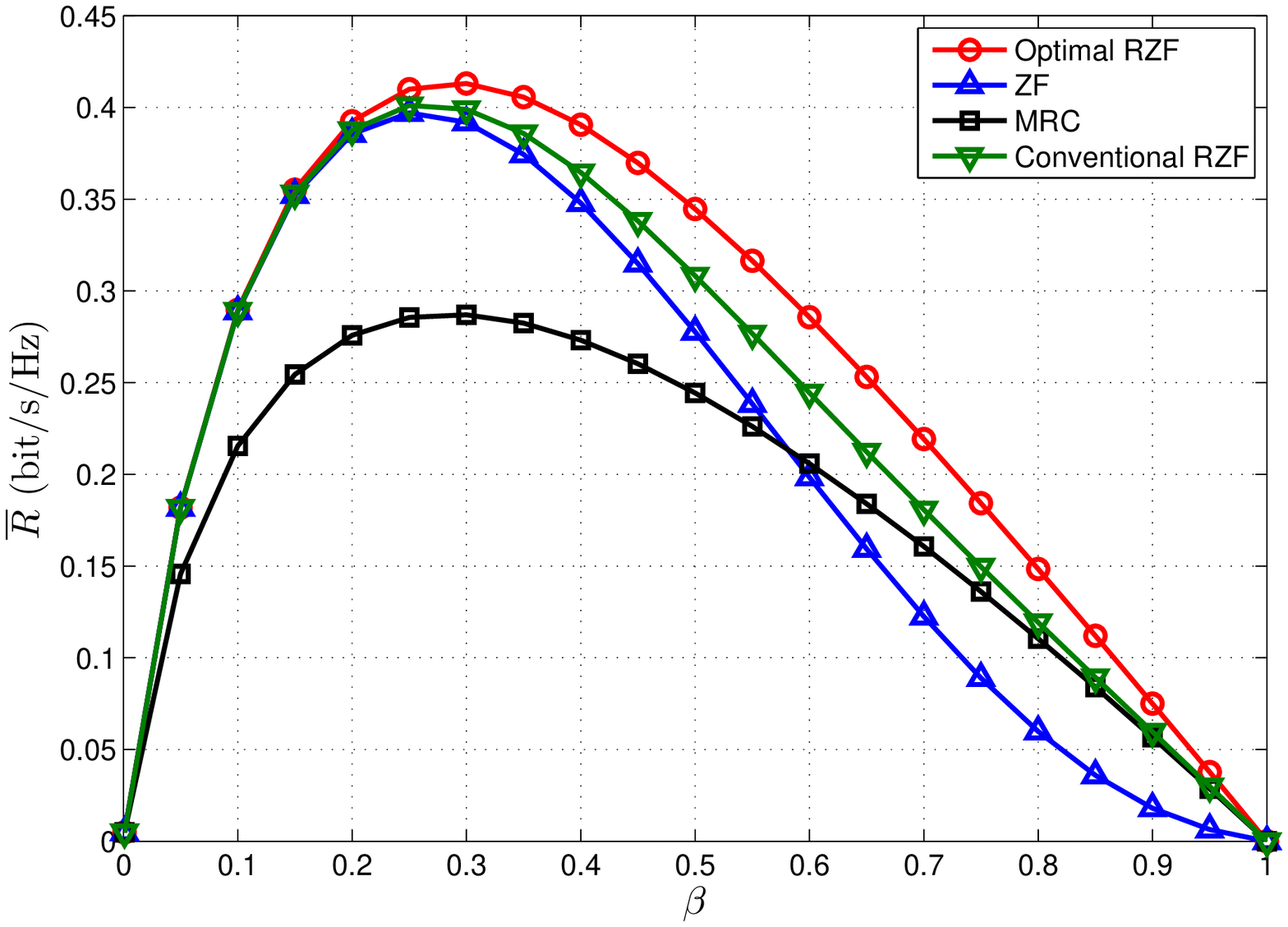}
\caption{Sum rate per antenna versus user loading ratio with optimal and conventional RZF, ZF, and MRC precoders ($\eta=1$, $b_{DA}=1$, $b_{AD}=3$, and $\nu=0.3$).}
\label{rate_beta_RZF_ZF_MRC_fig}
\end{figure}

\begin{figure}[tb]
\centering\includegraphics[width=0.6\textwidth]{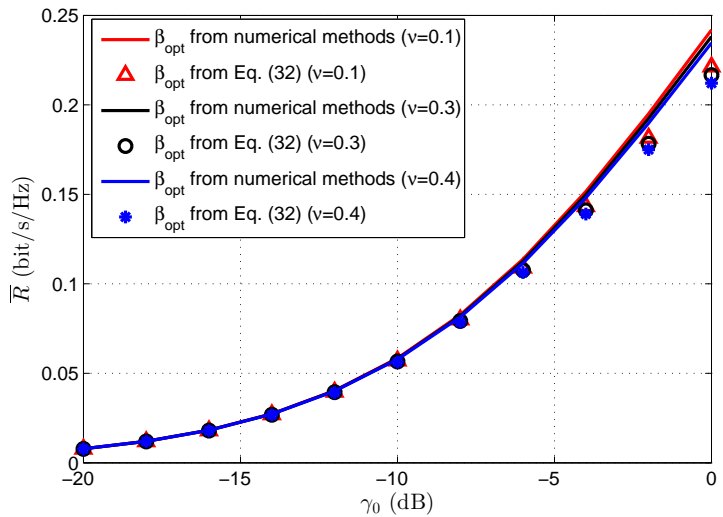}
\caption{Sum rate per antenna with $\beta^*$ under spatially correlated channels ($\eta=1$, $b_{DA}=1$, and $b_{AD}=3$).}
\label{Rate_beta_opt_cor}
\end{figure}

Fig. \ref{rate_beta_RZF_ZF_MRC_fig} compares the sum rate per antenna of the optimal and conventional RZF, ZF, and MRC precodings.
The SNR is $\gamma_0=10$~dB.
The optimal RZF precoding provides higher $\bar{R}$ than the other three precoders with $\beta$ ranging from 0 to 1, which implies that it is necessary to optimize the regularization paramater $\rho$.
With $\beta$ increasing from 0 to 1, $\bar{R}$ first increases since more users are served but then decreases due to higher pilot overhead.
Obviously, an optimal $\beta$ exists by maximizing $\bar{R}$.
Compared to the conventional RZF precoder, the optimal RZF precoder significantly improves the sum rate for $\beta\in(0.3,0.9)$.
This is because the optimal $\rho^*=\frac{ \left( \rho_{DA} \gamma_0 +1 \right) \beta}{(1-\rho_{DA})\gamma_0}$ differs more with the conventional $\rho=\frac{\beta}{\gamma_0}$ for larger $\beta$.
While for even larger $\beta\rightarrow 1$, we have $\overline{R}\rightarrow 0$ for $\eta=1$ and thus the performance gap tends small.

Although the closed-form $\beta^*$ in \eqref{beta_opt} is obtained for a special case without spatial correlation, Fig. \ref{Rate_beta_opt_cor} tests it for a massive MIMO channel with an exponential Toeplitz correlation matrix as in \eqref{r}.
Comparing the rates with the derived and exact $\beta^*$, we observe that the derived $\beta^*$ is still precise at low SNR with less significant correlations.
It implies that \eqref{beta_opt} can be applied to massive MIMO systems where the correlation is generally not severe.

\begin{figure}[tb]
\centering\includegraphics[width=0.6\textwidth]{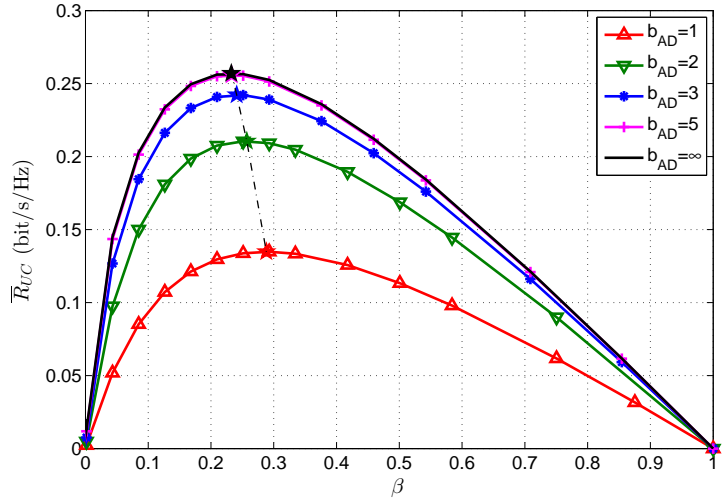}
\caption{Sum rate per antenna versus user loading ratio ($\gamma_0=0$ dB, $\eta=1$, and $b_{DA}=1$).}
\label{R_avg_beta_fig}
\end{figure}

\begin{figure}[tb]
\centering\includegraphics[width=0.6\textwidth]{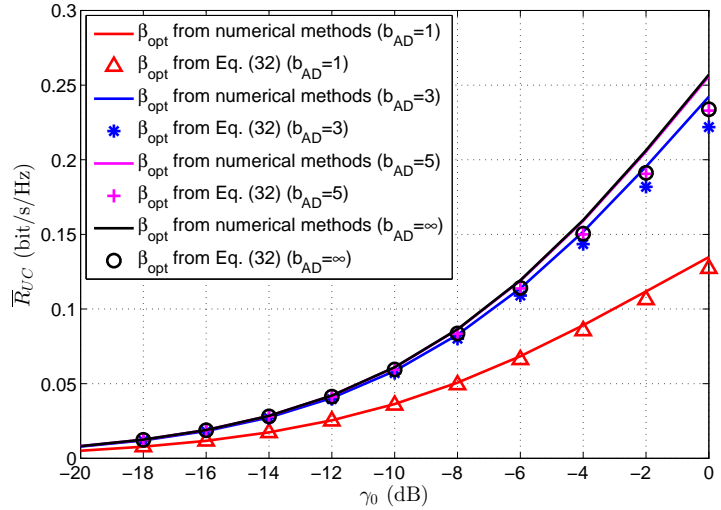}
\caption{Sum rate per antenna with $\beta^*$ under uncorrelated channels ($\eta=1$, $b_{DA}=1$).}
\label{R_avg_SNR_fig}
\end{figure}

We show the asymptotic rate performance without spatial correlation in the following.
Fig. \ref{R_avg_beta_fig} shows the sum rate per antenna under an uncorrelated channel.
The pentagram marks the maximum $\overline{R}_{UC}$ with $\beta^*$.
Specifically, $\beta^*$ equals $=0.2324,0.2330, 0.2409, 0.2570$, and $0.2881$, for $b_{AD}=\infty, 5,3,2$, and $1$, respectively.
We find that $\beta^*$ increases slightly when ADC resolution $b_{AD}$ decreases as discussed in Remark 3.
Fig. \ref{R_avg_SNR_fig} compares the performance of derived $\beta^*$ from \eqref{beta_opt} with exact value by numerical methods. It can be observed that $\overline{R}_{UC}$ with our derived $\beta^*$ is extremely close to that with exact one at low SNR.

\section{Conclusion}
In this paper, we investigate a downlink multiuser massive MIMO network and study the performance of low-resolution DACs equipped at BS and finite-resolution ADCs at use sides.
RZF precoding is used and transmit-side spatial correlation is considered.
By maximizing the asymptotic SIQNR, we derive the optimal regularization parameter which is found independent of the channel correlation and ADC quantization.
It implies that the optimal RZF precoding can be conducted at the BS without any knowledge of the ADC resolution at the user sides.
Moreover, by maximizing the sum rate per antenna, a closed-form solution of the optimal user loading ratio is obtained at low SNR.
Although the obtained $\beta^*$ is derived under a special case without correlation, it is verified also valid for slightly correlated channels.
Some extensions of this study, e.g., considering imperfect CSI and multi-antenna users with low-resolution ADCs, are of interest in future work.

\begin{appendices}

\section{}

The following lemma is recalled from \cite[Lemma 1]{cor0} for completeness.

\begin{lemma}
\label{lemma_ref}
In the large system limit, the following convergence of terms in \eqref{SIQNR} can be stated as follows:

1) Let $\mathbf{H}_k$ denote the channel matrix $\mathbf{H}$ removing the $k$th row $\mathbf{h}_{k}^{T}$ and we have
\begin{align}
\mathbf{h}_{k}^{T}\left(\mathbf{H}_k^H\mathbf{H}_k+\alpha\mathbf{I}_N\right)^{-1}\mathbf{h}^*_k \xrightarrow{a.s.} \xi,
\end{align}
where $\xi$ follows the same definition in \eqref{xi}.

2) Moreover,
\begin{align}
\mathbf{h}^T_k\left(\mathbf{H}_k^H\mathbf{H}_k+\alpha\mathbf{I}_N\right)^{-1}\mathbf{H}_k^H\mathbf{H}_k & \left(\mathbf{H}_k^H\mathbf{H}_k+\alpha\mathbf{I}_N\right)^{-1}\mathbf{h}^*_k
\xrightarrow{a.s.}
\frac{\beta E_{22}}{1-\beta E_{22}},
\end{align}
where $E_{22}$ follows the same definition in \eqref{E}.

3) As for constant $c$, it converges to
\begin{align}
c^2
\xrightarrow{a.s.}
-\frac{P(1+\xi)^2}{\beta \frac{\partial \xi}{\partial \rho}}.
\end{align}
\end{lemma}

\section{}
\begin{lemma}
\label{lemma_PP}
Considering a typical positive definite Hermitian correlation matrix $\mathbf{R}$ with a Toeplitz structure, $\textrm{diag}\left(\mathbf{PP}^H\right)$ converges to
\begin{align}
\label{asPP}
\textrm{diag}\left(\mathbf{PP}^H\right)
\xrightarrow{a.s.}
\frac{P}{N}\mathbf{I}_N.
\end{align}
\end{lemma}

\begin{proof}
Using \eqref{rzf}, we have
\begin{align}
\textrm{diag}\left(\mathbf{PP}^H \right)
& = c^2 \textrm{diag}\left(\left(\mathbf{H}^H\mathbf{H}+\alpha\mathbf{I}_N\right)^{-1}\mathbf{H}^H \mathbf{H} \left(\mathbf{H}^H\mathbf{H}+\alpha\mathbf{I}_N\right)^{-1} \right)
\nonumber
\\
& = c^2 \textrm{diag}\left(\left(\mathbf{H}^H\mathbf{H}+\alpha\mathbf{I}_N\right)^{-1}\left( \mathbf{H}^H \mathbf{H}+\alpha\mathbf{I}_N-\alpha\mathbf{I}_N \right)
 \left(\mathbf{H}^H\mathbf{H}+\alpha\mathbf{I}_N\right)^{-1} \right)
\nonumber
\\
& = c^2 \textrm{diag}\left(\left(\mathbf{H}^H\mathbf{H}+\alpha\mathbf{I}_N\right)^{-1}-  \alpha \left(\mathbf{H}^H\mathbf{H}+\alpha\mathbf{I}_N\right)^{-2} \right)
\nonumber
\\
& \xrightarrow{a.s.} c^2 \textrm{diag} \left(\left(M\mathbf{R}+\alpha\mathbf{I}_N\right)^{-1}
  - \alpha \left(M\mathbf{R}+\alpha\mathbf{I}_N\right)^{-2} \right)
\label{asHH}
\\
& \triangleq  d \mathbf{I}_N
\label{asDiag}
,
\end{align}
where \eqref{asHH} uses \eqref{H} and the fact that
$\frac{1}{M}\tilde{\mathbf{H}}^H\tilde{\mathbf{H}}-\mathbf{I} \xrightarrow{a.s.} \mathbf{0} $ due to the Central Limit Theorem.
According to the Continuous Mapping Theorem \cite{theorem}, the convergence preserves for continuous matrix functions like, $(\cdot)^{-1}$ and $\textrm{diag}(\cdot)$, of nonsingular matrices with positive definite Hermitian matrix $\mathbf{R}$ and $\alpha > 0$.
The matrix within the function $\textrm{diag}(\cdot)$ in \eqref{asHH} is still a Toeplitz matrix since the inverse of a positive definite Toeplitz matrix is asymptotically still Toeplitz \cite[Theorem 4.3]{theorem}.
Thus, it becomes a scaled identity matrix after taking the diagonal entries as in \eqref{asDiag}, where $d$ is defined as a constant (absorbing $c$) which can thus be determined by checking the power constraint of the term within the brackets, i.e., $\textrm{Tr}\left\{\mathbf{P}\mathbf{P}^{H}\right\}=P$, yielding $d = \frac{P}{N}$ in \eqref{asPP}.
\end{proof}

\section{Proof of Theorem~\ref{theorem_gamma}}

The proof of Theorem~\ref{theorem_gamma} applies Lemmas~\ref{lemma_ref} and \ref{lemma_PP}.
The value of \eqref{SIQNR} is calculated term-by-term. First concerning the numerator in \eqref{SIQNR}, the energy of the desired signal follows
\begin{align}
\left|\mathbf{h}_{k}^{T}\mathbf{p}_k\right|^2 
\label{S_F}
& = c^2\left|\mathbf{h}_{k}^{T}\left(\mathbf{H}^H\mathbf{H}+\alpha\mathbf{I}_N\right)^{-1}\mathbf{h}^*_k\right|^2
\\
\label{S_d}
& =   c^2 \frac{\left|\mathbf{h}_{k}^{T}\left(\mathbf{H}_k^H\mathbf{H}_k+\alpha\mathbf{I}_N\right)^{-1}\mathbf{h}^*_k\right|^2}{\left(1+\mathbf{h}_{k}^{T}\left(\mathbf{H}_k^H\mathbf{H}_k+\alpha\mathbf{I}_N\right)^{-1}\mathbf{h}^*_k\right)^2}
\\
\label{S_L13}
& \xrightarrow{a.s.} -\frac{P \xi^2}{\beta \frac{\partial \xi}{\partial \rho}}
,
\end{align}
where \eqref{S_F} is obtained by substituting \eqref{rzf}, \eqref{S_d} uses the matrix inversion lemma that

\begin{align}
\big(\mathbf{H}^H\mathbf{H}+\alpha\mathbf{I}_N\big)^{-1}\mathbf{h}^*_k
=&\big(\mathbf{H}_k^H\mathbf{H}_k+\alpha\mathbf{I}_N+\mathbf{h}^*_k\mathbf{h}^T_k\big)^{-1}\mathbf{h}^*_k
\nonumber\\
=&\big(\mathbf{H}_k^H\mathbf{H}_k+\alpha\mathbf{I}_N\big)^{-1}\mathbf{h}^*_k
-\frac{\big(\mathbf{H}_k^H\mathbf{H}_k+\alpha\mathbf{I}_N\big)^{-1} \mathbf{h}^*_k\mathbf{h}^T_k \big(\mathbf{H}_k^H\mathbf{H}_k+\alpha\mathbf{I}_N\big)^{-1}\mathbf{h}^*_k}{1+\mathbf{h}^T_k\big(\mathbf{H}_k^H\mathbf{H}_k+\alpha\mathbf{I}_N\big)^{-1}\mathbf{h}^*_k}
\nonumber\\
=&\frac{\big(\mathbf{H}_k^H\mathbf{H}_k+\alpha\mathbf{I}_N\big)^{-1}\mathbf{h}^*_k}{1+\mathbf{h}^T_k\big(\mathbf{H}_k^H\mathbf{H}_k+\alpha\mathbf{I}_N\big)^{-1}\mathbf{h}^*_k},
\label{H_d}
\end{align}
and \eqref{S_L13} utilizes 1) and 3) in Lemma~\ref{lemma_ref}, applying that convergence preserves for the continuous function according to the Continuous Mapping Theorem \cite{theorem}.
According to \eqref{xi} and \eqref{E}, we have
\begin{align}
\label{xi_d}
\frac{\partial \xi}{\partial \rho}=-\frac{(1+\xi)^2E_{12}}{1-\beta E_{22}}.
\end{align}
Substituting \eqref{xi_d} into \eqref{S_L13}, the desired signal power converges to
\begin{align}
\left|\mathbf{h}_{k}^{T}\mathbf{p}_k\right|^2
\xrightarrow{a.s.}
\frac{P\xi^2(1-\beta E_{22})}{(1+\xi)^2\beta E_{12}}.
\label{S}
\end{align}

Then, we consider the multiuser interference in the denominator of \eqref{SIQNR}.
It follows that
\begin{align}
\sum_{j\neq k}\left|\mathbf{h}_{k}^{T}\mathbf{p}_j\right|^2
\label{I_F}
&= c^2 \sum_{j\neq k}\left|\mathbf{h}_{k}^{T}\left(\mathbf{H}^H\mathbf{H}+\alpha\mathbf{I}_N\right)^{-1}\mathbf{h}^*_j\right|^2
\\
\nonumber
&= c^2 \mathbf{h}_{k}^{T}\left(\mathbf{H}^H\mathbf{H}+\alpha\mathbf{I}_N\right)^{-1}\mathbf{H}^H_k \mathbf{H}_k \left(\mathbf{H}^H\mathbf{H}+\alpha\mathbf{I}_N\right)^{-1} \mathbf{h}_{k}^*
\\
\label{I_d}
&= c^2 \frac{\mathbf{h}^T_k\left(\mathbf{H}_k^H\mathbf{H}_k+\alpha\mathbf{I}_N\right)^{-1}\mathbf{H}_k^H\mathbf{H}_k\left(\mathbf{H}_k^H\mathbf{H}_k+\alpha\mathbf{I}_N\right)^{-1}\mathbf{h}^*_k}
{\left(1+\mathbf{h}^T_k\left(\mathbf{H}_k^H\mathbf{H}_k+\alpha\mathbf{I}_N\right)^{-1}\mathbf{h}^*_k\right)^2}
\\
\label{I_L23}
& \xrightarrow{a.s.} - \frac{P E_{22}}{(1-\beta E_{22})\frac{\partial \xi}{\partial \rho}}\\
\label{I_u}
&= \frac{PE_{22}}{(1+\xi)^2E_{12}}
,
\end{align}
where \eqref{I_F} is obtained by substituting \eqref{rzf}, \eqref{I_d} follows by \eqref{H_d}, \eqref{I_L23} utilizes Lemma~\ref{lemma_ref}, and \eqref{I_u} uses \eqref{xi_d}.

For the term of DAC quantization distortion, we substitute \eqref{cor_DA} and \eqref{H} and have
\begin{align}
\mathbf{h}_{k}^{T}\mathbb{E}\left\{\mathbf{n}_{DA}\mathbf{n}_{DA}^{H}\right\}\mathbf{h}_{k}^{*}
&=\rho_{DA}\tilde{\mathbf{h}}_{k}^{T} \mathbf{R}^{\frac{1}{2}}  \textrm{diag}\left(\mathbf{PP}^H\right) \mathbf{R}^{\frac{1}{2}} \tilde{\mathbf{h}}_{k}^{*}
\nonumber
\\
&\xrightarrow{a.s.} \rho_{DA} \frac{P}{N} \tilde{\mathbf{h}}_{k}^{T}   \mathbf{R}  \tilde{\mathbf{h}}_{k}^{*}
\label{PP}
\\
&\xrightarrow{a.s.} \rho_{DA} \frac{P}{N} \textrm{Tr}\{\mathbf{R}\}
\label{I_DA0}
\\
&= \rho_{DA} P,
\label{I_DA}
\end{align}
where $\tilde{\mathbf{h}}_{k}^{T}$ denotes the $k$th row of $\tilde{\mathbf{H}}$
and \eqref{PP} uses the asymptotical property of $\textrm{diag}\left(\mathbf{PP}^H\right)$ in Lemma~\ref{lemma_PP},
\eqref{I_DA0} uses \cite[Corollary 1]{matrix2} which implies that
\begin{align}
\frac{1}{\sqrt{N}} \tilde{\mathbf{h}}_{k}^{T}   \mathbf{R} \frac{1}{\sqrt{N}}  \tilde{\mathbf{h}}_{k}^{*} \xrightarrow{a.s.} \frac{1}{N} \textrm{Tr}\{\mathbf{R}\},
\end{align}
and \eqref{I_DA} comes from the normalization constraint that $\textrm{Tr}\{\mathbf{R}\}=N$.

While for the term of ADC quantization distortion in \eqref{SIQNR}, we substitute \eqref{DAC} and \eqref{y} into \eqref{cor_AD} and get
\begin{align}
\mathbb{E}\{|n_{AD,k}|^2\}
=&\rho_{AD}(1-\rho_{AD})
\left[(1-\rho_{DA})\mathbf{h}_{k}^{T}\mathbf{PP}^H \mathbf{h}_{k}^{*}\!+\!\mathbf{h}_{k}^{T}\mathbb{E}\left\{\mathbf{n}_{DA}\mathbf{n}_{DA}^{H}\right\}\mathbf{h}_{k}^{*}\!+\!\sigma_n^2\!\right]
\nonumber
\\
 =& \rho_{AD}(1-\rho_{AD})
\left[ (1-\rho_{DA}) \frac{\xi^2P+(1-\xi^2)\beta P E_{22}}{(1+\xi)^2\beta E_{12}}+\rho_{DA}P+\sigma_n^2\right],
\label{I_AD}
\end{align}
where \eqref{I_AD} is obtained by substituting \eqref{S}, \eqref{I_u} and \eqref{I_DA}.

Thus far, by substituting \eqref{S}, \eqref{I_u}, \eqref{I_DA} and \eqref{I_AD} into \eqref{SIQNR},  and using the following equality as
\begin{align}
\label{xi0}
\xi(1-\beta E_{22})=\rho(1+\xi)^2E_{12}+\beta E_{22},
\end{align}
which comes from
\begin{align}
\xi
&=\int \frac{(1+\xi)\lambda}{\rho(1+\xi)+\beta \lambda} \textrm{d}\Lambda(\lambda)
=\int \frac{\rho(1+\xi)^2\lambda+\beta \lambda^2+ \beta\xi \lambda^2}{[\rho(1+\xi)+\beta \lambda]^2}\textrm{d}\Lambda(\lambda)
\label{xi_2}
=\rho(1+\xi)^2E_{12}+\beta E_{22}+\beta\xi E_{22},
\end{align}
we finally complete the proof of Theorem 1.

\section{}
The following presents detailed derivations of $\xi$, $E_{12}$, and $E_{22}$ for simulation.
For the exponential Toeplitz correlation matrix $\mathbf{R}$ as given in \eqref{r}, we have \cite{Toeplitz}
\begin{align}
\label{toe}
\lim_{N\rightarrow\infty}\frac{1}{N}\sum_{n=1}^N F(\lambda_n) \xrightarrow{a.s.} \frac{1}{2\pi} \int_0^{2\pi} F(f(w)) \textrm{d} w,
\end{align}
where $\lambda_n~(n=1,2,...,N)$ denotes the eigenvalues of $\mathbf{R}$, $F(\cdot)$ denotes any continuous function in the support of $f(w)$, and $f(w)$ is the spectral density of $r_{ij}$ given by
\begin{align}
f(w)&=\lim_{N\rightarrow\infty} \sum_{n=-N+1}^{N-1} \nu^{|n|}e^{jnw}
=\frac{1-\nu^2}{1-2\nu\cos w+\nu^2}.
\label{f}
\end{align}
Using \eqref{toe} to evaluate the expectation in \eqref{xi}, $\xi$ can be rewritten as
\begin{align}
\xi
&=\mathbb{E}_{\lambda} \left \{ \frac{\lambda(1+\xi)}{\rho(1+\xi)+\beta\lambda} \right \}
\nonumber
\\
&=\lim_{N\rightarrow\infty}\frac{1}{N}\sum_{n=1}^N \frac{\lambda_n(1+\xi)}{\rho(1+\xi)+\beta\lambda_n}
\nonumber
\\
&\xrightarrow{a.s.}\frac{1}{2\pi} \int_0^{2\pi} \frac{f(w)(1+\xi)}{\rho(1+\xi)+\beta f(w)} \textrm{d}w
\nonumber
\\
&= \frac{1-\nu^2}{2\pi} \int_0^{2\pi} \frac{1}{a+b\cos w} \textrm{d}w
\label{xi_toe_f}
\\
&=\frac{1-\nu^2}{\sqrt{a^2-b^2}},
\label{xi_toe}
\end{align}
where \eqref{xi_toe_f} is obtained by substituting \eqref{f} and we define
\begin{align}
\label{a}
a\triangleq\rho(1+\nu^2)+\frac{\beta (1-\nu^2)}{1+\xi},
\end{align}
and
\begin{align}
\label{b}
b\triangleq -2\rho\nu.
\end{align}
\eqref{xi_toe} comes from \cite[Eq. 3.661 (4)]{table}, which is written as follows
\begin{align}
&\frac{1}{2}\int_0^{2\pi} \frac{1}{(a+b\cos x)^{n+1}} \textrm{d}x
\!=\!\frac{\pi}{2^n(a\!+\!b)^n\sqrt{a^2-b^2}}\!\sum_{k=0}^n\!\frac{(2n\!-\!2k\!-\!1)!!(2k\!-\!1)!!}{(n\!-\!k)!k!}\left(\!\frac{a\!+\!b}{a\!-\!b}\!\right)\!^k
\label{table3661}
.
\end{align}
Note that since $\xi$ is also involved in $a$, the calculation of $\xi$ still cannot be directly solved through \eqref{xi_toe} in closed form.
Fortunately, we can resort to numerical methods to solve the equalities \eqref{xi_toe}, \eqref{a}, and \eqref{b}.

As for $E_{12}$, using \eqref{E} and \eqref{toe}, it yields
\begin{align}
E_{12}
&=\frac{1}{2\pi} \int_0^{2\pi} \frac{f(w)}{[\rho(1+\xi)+\beta f(w)]^2 }\textrm{d}w
\nonumber
\\
&=\frac{1-\nu^2}{2\pi(1+\xi)^2} \int_0^{2\pi} \frac{A+B\cos w}{(a+b\cos w)^2} \textrm{d}w
\label{E12_toe_f}
\\
&=\frac{(1-\nu^2)(aA-bB)}{2\pi(1+\xi)^2(a^2-b^2)} \int_0^{2\pi} \frac{1}{a+b\cos w} \textrm{d}w
\label{E12_toe_2}
\\
&=\frac{(1-\nu^2)(aA-bB)}{(1+\xi)^2(a^2-b^2)^{\frac{3}{2}}}
\label{E12_toe}
,
\end{align}
where \eqref{E12_toe_f} uses \eqref{f} and the definitions $A\triangleq1+\nu^2$ and $B\triangleq-2\nu$,
\eqref{E12_toe_2} comes from \cite[Eq. 2.554 (1)]{table}, which is rewritten as
\begin{align}
&\int \frac{A\!+\!B\cos x}{(a\!+\!b\cos x)^n} \textrm{d}x
\!=\!\frac{1}{(\!n\!-\!1\!)(\!a^2\!-\!b^2\!)} \left[\frac{(\!aB\!-\!Ab\!)\sin x}{(a\!+\!b\cos x)^{n-1}}
\!+\! \int\frac{(\!Aa\!-\!bB\!)(n\!-\!1)\!+\!(\!n\!-\!2\!)(aB\!-\!bA)\cos x}{(a+b\cos x)^{n-1}} \textrm{d}x \right]
\label{table2554}
,
\end{align}
and \eqref{E12_toe} utilizes \eqref{table3661}.
Similarly, $E_{22}$ can be derived as 
\begin{align}
E_{22}
&=\frac{1}{2\pi} \int_0^{2\pi} \frac{f^2(w)}{[\rho(1+\xi)+\beta f(w)]^2 }\textrm{d}w
\nonumber
\\
&=\frac{1}{2\pi} \int_0^{2\pi} \frac{\left(\frac{1-\nu^2}{1-2\nu\cos w+\nu^2}\right)^2}{\left[\rho(1+\xi)+\beta \frac{1-\nu^2}{1-2\nu\cos w+\nu^2}\right]^2 }\textrm{d}w
\label{E22_f}\\
&=\frac{1}{2\pi} \int_0^{2\pi} \frac{(1-\nu^2)^2}{\left[\rho(1\!+\!\xi)(1\!-\!2\nu\cos w\!+\!\nu^2)\!+\!\beta(1\!-\!\nu^2)\right]^2} \textrm{d}w
\nonumber\\
&=\!\frac{(1-\nu^2)^2}{2\pi(1\!+\!\xi)^2} \int_0^{2\pi}\!\!\!
 \frac{1}{\left[\!-\!2\rho\nu\cos w\!+\!\rho(1\!+\!\nu^2)\!+\!\frac{\beta(1\!-\!\nu^2)}{1\!+\!\xi}\right]^2} \textrm{d}w
\nonumber\\
&\triangleq\frac{(1-\nu^2)^2}{2\pi(1+\xi)^2} \int_0^{2\pi} \frac{1}{(a+b\cos w)^2} \textrm{d}w
\label{E22_toe_f}
\\
&=\frac{a(1-\nu^2)^2}{2\pi(1+\xi)^2(a^2-b^2)} \int_0^{2\pi} \frac{1}{a+b\cos w} \textrm{d}w
\label{E22_toe_2}
\\
&=\frac{a(1-\nu^2)^2}{(1+\xi)^2(a^2-b^2)^{\frac{3}{2}}}
\label{E22_toe}
,
\end{align}
where \eqref{E22_f} is obtained by utilizing $f(w)$ given in \eqref{f}, \eqref{E22_toe_f} uses the definitions of $a$ and $b$ in \eqref{a} and \eqref{b} respectively,
\eqref{E22_toe_2} uses \cite[Eq. 2.554 (3)]{table}, which is rewritten as
\begin{align}
&\int \frac{1}{(a+b\cos x)^{n}} \textrm{d}x
=-\frac{1}{(n-1)(a^2-b^2)}
\left[ \frac{b\sin x}{(a+b\cos x)^{n-1}}-\int \frac{(n-1)a-(n-2)b\cos x}{(a+b\cos x)^{n-1}} \textrm{d}x \right]
\label{table2554_3}
,
\end{align}
and \eqref{E22_toe} comes from \eqref{table3661}.

\end{appendices}







\begin{thebibliography}{1}
\bibitem{part} J. Xu, W. Xu, F. Shi, and H. Zhang, ``User loading in downlink multiuser massive MIMO with 1-bit DAC and quantized receiver,'' in \emph{Proc. IEEE VTC}, Toronto, Canada, Sept. 2017.

\bibitem{MIMO0}  T. L. Marzetta, ``Noncooperative cellular wireless with unlimited numbers of base station antennas,'' \emph{IEEE Trans. Wirel. Commun.,} vol.~9, no.~11, pp.~3590--3600, Nov. 2010.

\bibitem{MIMO1} E. G. Larsson, O. Edfors, F. Tufvesson, and T. L. Marzetta, ``Massive MIMO for next generation wireless systems,'' \emph{IEEE Commun. Mag.}, vol.~52, no.~2, pp.~186--195, Feb. 2014.

\bibitem{MIMO3} J. Hoydis, S. Brink, and M. Debbah, ``Massive MIMO in the UL/DL of cellular networks: How many antennas do we need?'', \emph{IEEE J. Sel. Areas Commun.}, vol.~31, no.~2, pp.~160--171, Feb. 2013.

\bibitem{MIMO2} H. Xie, F. Gao, and S. Jin, ``An overview of low-rank channel estimation for massive MIMO systems,'' \emph{IEEE Access}, vol.~4, pp.~7313--7321, Nov. 2016.


\bibitem{DAC_PSK} H. Jedda, A. Mezghani, J. A. Nossek, and A. L. Swindlehurst, ''Massive MIMO downlink 1-bit precoding with linear programming for PSK signaling,'' in \emph{Proc. IEEE SPAWC}, Sapporo, Japan, July 2017.

\bibitem{DAC_pokemon} O. Casta$\tilde{\textrm{n}}$eda, T. Goldstein, and C. Studer, ``POKEMON: A non-linear beamforming algorithm for 1-bit massive MIMO,'' in \emph{Proc. IEEE ICASSP}, New Orleans, LA, USA, June 2017.

\bibitem{DAC_pert} A. L. Swindlehurst, A. K. Saxena, A. Mezghani, and I. Fijalkow, ``Minimum probability-of-error perturbation precoding for the one-bit massive MIMO downlink,'' in \emph{Proc. IEEE ICASSP}, New Orleans, LA, USA, June 2017.


\bibitem{DAC_downlink} Y. Li, C. Tao, A. L. Swindlehurst, A. Mezghani, and L. Liu, ``Downlink achievable rate analysis in massive MIMO systems with one-bit DACs,'' \emph{IEEE Commun. Lett.}, vol. 21, no. 7, pp. 1669--1672, July 2017.


\bibitem{DAC1} A. K. Saxena, I. Fijalkow, and A. L. Swindlehurst, ``Analysis of one-bit quantized precoding for the multiuser massive MIMO downlink,'' \emph{IEEE Trans. Sig. Process.}, vol.~65, no.~17, pp.~4624--4634, Sept. 2017.

\bibitem{DAC2} S. Jacobsson, G. Durisi, M. Coldrey, T. Coldstein, and C. Studer, ``Quantized precoding for massive MU-MIMO,'' \emph{IEEE Trans. Commun.,} vol.~65, no.~11, pp.~4670--4684, July 2017.


\bibitem{ADC1} C. Risi, D. Persson, and E. G. Larsson, ``Massive MIMO with 1-bit ADC,'' Apr. 2014, [Online]. Available:
http://arxiv.org/abs/1404.7736

\bibitem{ADC2} J. Liu, J. Xu, W. Xu, S. Jin, and X. Dong, ``Multiuser massive MIMO relaying with mixed-ADC receiver,'' \emph{IEEE Sig. Process. Lett.}, vol.~24, no.~1, pp. 76--80, Dec. 2016.

\bibitem{ADC} J. Zhang, L. Dai, S. Sun, and Z. Wang, ``On the spectral efficiency of massive MIMO systems with low-resolution ADCs,'' \emph{IEEE Commun. Lett.,} vol.~20, no.~5, pp. 842--845, May 2016.

\bibitem{ADC Mix} H. Pirzadeh and A. L. Swindlehurst, ``Spectral efficiency under energy constraint for mixed-ADC MRC massive MIMO,'' \emph{IEEE Sig. Process. Lett.}, vol.~24, no.~12, pp.~1847--1851, Dec. 2017.



\bibitem{load} W. Dai, Y. Liu, B. Rider, and W. Gao, ``How many users should be turned on in a multi-antenna broadcast channel?'' \emph{IEEE J. Sel. Areas Commun.}, vol.~26, no.~8, pp.~1526--1535, Oct. 2008.

\bibitem{load1} S. Park, J. Park, A. Y. Panah, and R. W. Heath Jr., ``Optimal user loading in massive MIMO systems with regularized zero forcing precoding,'' \emph{IEEE Wirel. Commun. Lett.}, vol. 6, no. 1, pp. 118--121, Feb. 2017.
\bibitem{load2} R. Muharar, R. Zakhour, and J. Evans. ``Optimal power allocation and user loading for multiuser MISO channels with regularized channel inversion,'' \emph{IEEE Trans. Commun.}, vol. 61, no. 12, pp. 5030--5041, Dec.~2013.



\bibitem{load4} Y. Xie, B. Li, J. Fan, X. Zhou, G. Y. Li, and X. Li, ``User grouping with load balance in FDD massive MIMO systems,'' in \emph{Proc. IEEE VTC}, Toronto, Canada, Sept. 2017.

\bibitem{cor5} D. Chizhik, F. Rashid-Farrokhi, J. Ling, and A. Lozano, ``Effect of antenna separation on the capacity of BKAST in correlated channels,'' \emph{IEEE Commun. Lett.,} vol.~4, no.~11, pp.~337--339, Nov. 2000.

\bibitem{cor6} C.-N. Chuah, D. Tse, J. M. Kahn, and R. A. Valenzuela, ``Capacity scaling in MIMO wireless systems under correlated fading,'' \emph{IEEE Trans. Inf.
    Theory,} vol.~48, no.~3, pp. 637--650, Mar. 2002.






\bibitem{cor7} A. M. Tulino, A. Lozano, and S. Verdu, ``Impact of antenna correlation on the capacity of multiantenna channels,'' \emph{IEEE Trans. Inf. Theory,} vol.~51, no.~7, pp.~2491--2509, July~2005.

\bibitem{cor3} X. Mestre, J. R. Fonollosa, and A. Pages-Zamora, ``Capacity of MIMO channels: Asymptotic evaluation under correlated fading,'' \emph{IEEE J. Sel. Areas Commun.}, vol.~21, no.~5, pp.~829--838, June 2003.

\bibitem{cor4} C. Martin and B. Ottersten, ``Asymptotic eigenvalue distributions and capacity for MIMO channels under correlated fading,'' \emph{IEEE Trans. Wirel. Commun.}, vol.~3, no.~4, pp.~1350--1359, July 2004.

\bibitem{cor0} R. Muharar and J. Evans, ``Downlink beamforming with transmit-side channel correlation: A large system analysis,'' in \emph{Proc. IEEE ICC}, Kyoto, Japan, June 2011.

\bibitem{cor2} P. Dong, H. Zhang, W. Xu, G. Y. Li, and X. You, ``Performance analysis of multiuser massive MIMO with spatially correlated channels using low-precision ADC,'' \emph{IEEE Commun. Lett.}, vol.~22, no.~1, pp.~205--208, Jan.~2018.

\bibitem{CE_1} J. Mo, P. Schniter, N. Prelcic, and R. Heath, ``Channel estimation in millimeter wave MIMO systems with one-bit quantization,'' in \emph{Proc. ASILOMAR,} Pacific Grove, CA, USA, Nov. 2014.

\bibitem{CE_2} L. Fan, D. Qiao, S. Jin, C. Wen, and M. Matthaiou, ``Optimal pilot length for uplink massive MIMO systems with low-resolution ADC,'' in \emph{Proc. IEEE SAM,} Rio de Janeiro, Brazil, July 2016.

\bibitem{CE_3} H. Xie, F. Gao, and S. Jin, ``An overview of low-rank channel estimation for massive MIMO systems,'' \emph{IEEE Access}, vol.~4, pp. 7313--7321, Nov. 2016.



\bibitem{Bus1} J. J. Bussgang, ``Crosscorrelation functions of amplitude-distorted Gaussian signals,'' Res. Lab. Elec., Cambridge, MA, Tech. Rep. 216, Mar.~1952.

\bibitem{Bus2} A. Mezghani and J. Nossek, ``Capacity lower bound of MIMO channels with output quantization and correlated noise,'' in \emph{Proc. IEEE ISIT}, Cambridge, MA, USA, July 2012.


\bibitem{ADC3} J. Xu, W. Xu, and F. Gong, ``On performance of quantized transceiver in multiuser massive MIMO downlinks,'' \emph{IEEE Wirel. Commun. Lett.,} vol.~6, no. 5, pp. 562--565, June 2017.


\bibitem{Mixed} J. Zhang, L. Dai, Z. He, B. Ai, and O. A. Dobre, ``Mixed-ADC/DAC multipair massive MIMO relaying systems: Performance analysis and power optimization,'' \emph{IEEE Trans. Commun.}, to appear.



\bibitem{cor1} D.-S. Shiu, G. J. Foschini, M. J. Gans, and J. M. Kahn, ``Fading correlation and its effect on the capacity of multielement antenna systems,'' \emph{IEEE Trans. Commun.}, vol.~48, no.~3, pp. 502--513, Mar. 2000.


\bibitem{Power} E. Bj\"{o}rnson, E. G. Larsson, and M. Debbah, ``Massive MIMO for maximal spectral efficiency: How many users and pilots should be allocated?'' \emph{IEEE Trans. Wirel. Commun.}, vol.~15, no.~2, pp.~1293--1308, Feb. 2016.


\bibitem{rho} J. Max, ``Quantizing for minimum distortion,'' \emph{IRE Trans. Inform. Theory}, vol. 6, no. 1, pp. 7--12, 1960.

\bibitem{matrix} A. M. Tulino and S. Verdu, \emph{Random Matrix Theory and Wireless Communications}, Hanover, MA: Now Publishers Inc., 2004.

\bibitem{load3} V. K. Nguyen and J. S. Evans, ``Multiuser transmit beamforming via regularized channel inversion: A large system analysis,'' in \emph{Proc. IEEE GLOBECOM,} New Orleans, LA, Dec. 2008.

\bibitem{theorem} P. Billingsley, \emph{Convergence of Probability Measures}, Hoboken, YK: John Wiley \& Sons, 1969.


\bibitem{matrix2} J. Evans and D. N. C. Tse, ``Large system performance of linear multiuser receivers in multipath fading channels,'' \emph{IEEE Trans. Inform. Theory}, vol. 46, no. 6, pp. 2059--2078, Sept. 2000.

\bibitem{Toeplitz} R. M. Gary, ``Toeplitz and Circulant Matrices: A Review,'' \emph{Foundations and Trends in Commun. and Inform. Theory,} vol.~2, no.~3, 2005.


\bibitem{table} I. S. Gradshteyn and I. M. Ryzhik, \emph{Table of Integrals, Series, and Products}, Burlington, MA: Academic Press, 2014.





\end{thebibliography}
%

\end{document}